\documentclass[%
superscriptaddress,
nofootinbib,
 amsmath,amssymb,
 aps,
]{revtex4-2}

\usepackage{uniinput}
\usepackage{comment}
\usepackage[english]{babel}
\usepackage{graphicx} 
\usepackage{dcolumn} 
\usepackage{bm}
\usepackage{bbm}
\usepackage{bbold}
\usepackage{xcolor}
\usepackage{multirow}
\usepackage{booktabs}
\usepackage{array}
\usepackage{makecell}
\usepackage{siunitx} 
\usepackage{hyperref}
\usepackage[capitalise]{cleveref}
\usepackage{amsthm}
\newtheorem{proposition}{Proposition}
\newtheorem{proposition_restatement}{Proposition}

\newcommand{\Tr}{\operatorname{Tr}}
\newcommand{\norm}[1]{\left\lVert#1\right\rVert}
\newcommand{\tnorm}[1]{\norm{#1}_1}
\newcommand{\Lind}{\mathcal{L}}
\newcommand{\re}{\rm{Re}\,}
\newcommand{\im}{\rm{Im}\,}
\newcommand{\ad}{\text{ad}}
\DeclareRobustCommand{\circled}[1]{\ensuremath{\text{\tikz[baseline=(c.base)]
    \node[draw, circle, inner sep=0.6pt, minimum size=1.35em, line width=0.4pt] (c)
    {\scriptsize #1};}}}
\AtBeginDocument{\pdfstringdefDisableCommands{\def\circled#1{(#1)}}}

\usepackage{xcolor}
\definecolor{riverlane_green}{RGB}{0, 111, 98}
\definecolor{riverlane_light_green}{RGB}{0, 150, 143}
\definecolor{riverlane_orange}{RGB}{255, 117, 0}
\definecolor{riverlane_red}{RGB}{220, 68, 5}
\definecolor{riverlane_pink}{RGB}{207, 111, 127}
\usepackage{hyperref}
\hypersetup{
  colorlinks   = true, 
  urlcolor     = riverlane_green, 
  linkcolor    = riverlane_green, 
  citecolor   = riverlane_green  
}

\usepackage{tikz}
\usetikzlibrary{quantikz} 

\begin{document}

\title{Fault-tolerant resource estimation for ground-state preparation via Lindblad simulation}

\author{Marius Bothe$^*$}
\affiliation{Riverlane, Cambridge, CB2 3BZ, UK}
\affiliation{Astex Pharmaceuticals, 436 Cambridge Science Park, Cambridge, CB4 0QA, UK}
\author{Patrick Schoepf}
\affiliation{Astex Pharmaceuticals, 436 Cambridge Science Park, Cambridge, CB4 0QA, UK}
\date{\today}
\author{Nick S. Blunt}
\affiliation{Riverlane, Cambridge, CB2 3BZ, UK}

\begin{abstract}
Recent advances in algorithms for simulating Lindblad dynamics have clarified their theoretical potential for state preparation, but their practicality in the early fault‑tolerant regime and beyond remains uncertain. In this work, we address this by investigating the cost of preparing ground states of the fermionic Hubbard model, following the single-ancilla approach of [Phys. Rev. Research 6, 033147 (2024)]. We derive rigorous error bounds including constant prefactors, and compare to empirical error behavior and practical convergence parameters obtained from circuit‑level simulations. We find that empirically chosen parameters can reduce the required resources by orders of magnitude, resulting in an estimated $7.7 \times 10^{8}$ T gates to perform one unit of time evolution targeting the low-energy subspace of a 36-site fermionic Hubbard model. We identify the accurate filtering of energy transitions as the main source of this cost. 
\end{abstract}

\maketitle

\def\thefootnote{$^*$}\footnotetext{\href{marius.bothe@riverlane.com}{marius.bothe@riverlane.com}}
\tableofcontents

\section{Introduction}
\label{sec:intro}

In recent years there have been various algorithmic advances for the simulation of Lindblad dynamics on quantum computers: More efficient Lindblad simulation techniques \cite{yu25twoancilla, leng2026accelerating, elman2026localityfilterdesigndissipative,mizuta2026lattice, gilyen2026metropolis,shang2026exponential,kato2026},  exact detailed-balance constructions \cite{chen2023efficient, chen2023quantum}, compact ground-state preparation routines \cite{ding2024single,ding2025end}, comparisons of detailed-balance conditions \cite{ding2024efficient} and various exact and numerical results on convergence times \cite{vsmid2501polynomial, tong2025fast, zhan2026rapid}. An open question remains how practically feasible Lindblad evolution will be as an algorithm in the early fault-tolerant regime and beyond. To evaluate this, in this paper we focus on the single-ancilla ground state preparation algorithm by Ding et al. \cite{ding2024single}, which we choose for three reasons: firstly, ground-state preparation is a common problem that is also a subroutine for many other algorithms such as quantum phase estimation, so it is widely applicable. Secondly, the algorithm uses the lower-overhead dilation construction of \cite{cleve1612efficient}, which is likely to be more suitable for the tight resource constraints of the early-fault-tolerant setting.
Finally, the algorithm uses a straightforward construction to ensure convergence to the ground state, which is less resource intensive than some of the exact detailed-balance constructions of \cite{chen2023efficient}. Indeed, the convergence of the single-ancilla algorithm has been studied in detail in \cite{zhan2026rapid}, where it was shown for many systems of interest that rapid mixing (i.e., logarithmic scaling of the mixing time with the system size) can be achieved, highlighting the promise of the algorithm for state preparation tasks. However, the authors stop short of determining the algorithmic parameters and compiling down to fault-tolerant resources, leaving the question of practical applications open.

To investigate this, we set out to estimate the resource requirements of this algorithm: we extend the original asymptotic results of \cite{ding2024single} by deriving the full prefactors for the scaling. With this we implemented resource estimation in Qualtran \cite{harrigan2024expressing} to fully cost the gate counts for performing Lindblad evolution for a given duration with rigorous error guarantees.

Furthermore, we perform simulations for the one-dimensional Hubbard model, whose generalizations are a leading candidate for early fault-tolerant applications \cite{campbell2022early,bay2025fault,akahoshi2025compilation,bay2026quantum}. These circuit-level simulations for small system sizes allow us to make empirical estimates for simulation errors, which are typically orders of magnitude smaller than the rigorous bounds. It also allows estimates for practical convergence times, which are difficult to prove otherwise in the zero-temperature regime for interacting fermionic systems. 
To further reduce costs, we investigate the choice of mixing operators and optimal filter functions, which ensure ergodicity and fast convergence while respecting particle number and spin symmetry. Finally, we evaluate these gate counts in the Pauli-based computation model \cite{bravyi2016trading, litinski2019game}. 

In \cref{sec:algorithm} we introduce the single-ancilla ground-state algorithm and discuss the rigorous error bounds that we derive. The total error depends on the total simulation time, which for state preparation should be chosen to be (some multiple of) the mixing time of the Lindblad evolution. Therefore, to guarantee a certain accuracy in the final state, the parameters have to be chosen with prior knowledge of the mixing time. Despite many recent results showing fast (polynomial) or even rapid (logarithmic) mixing times for Lindbladian constructions \cite{vsmid2501polynomial, ding2025end,tong2025fast, fang2025mixing,bergamaschi2025quantum} there are so far no results guaranteeing fast mixing for any classically challenging quantum system. This is not surprising, since such results rarely exist for classical sampling algorithms either, yet many classical algorithms, such as Markov chain sampling, are widely successful nonetheless. 

To investigate the practical usefulness of Lindblad state preparation we therefore turn to simulation results in \cref{sec:simulation}. For this, we implement a circuit-level simulation of the ground state preparation algorithm \cite{ding2024single} for the fermionic Hubbard model, a widely studied candidate for early quantum advantage. This simulation allows us to simulate the effects of different parameter choices on the final-state error, measured in distance from the ground-state energy. We find a convergence speed versus precision trade-off, and further identify optimal parameter choices for convergence when the spectral gap of the system Hamiltonian is known. We also simulate the effect of device noise on the evolution, and find that in most relevant applications the contractiveness of the Lindblad evolution offers no protection against device noise, contrary to the idea of algorithm-level error correction.

In \cref{sec:costing} we estimate the cost of preparing a classically challenging ground state via Lindblad evolution to an energy precision that scales with system size. For this, we bound the spectral norm and estimate the gap of the Hubbard Hamiltonian\footnote{In practice, when targeting an extensive energy error, the gap estimate becomes irrelevant, because for all but the most gapped systems the allowed energy error becomes larger than the gap, so we use that instead.}, which we feed into the error bounds to find rigorous estimates for the parameters and gate counts required. 
We also extrapolate the parameters found to reach the required accuracy for small systems and scale them according to the asymptotic scaling laws to estimate the parameters required in practice. 
The fault-tolerant resources required for both the rigorous regime and the empirically-found parameters are then evaluated by estimating the non-Clifford gate cost \cite{bravyi2016trading, litinski2019game}. The resulting resource estimates show that even for a moderately-sized system the cost is substantial, requiring $7.7 \times 10^{8}$ T gates for one unit of time evolution for a 36-site fermionic Hubbard model, with the cost of filtering the energy transitions being the main bottleneck. We conclude that for practical usefulness we still require large improvements in algorithmic efficiency, similar to the development seen in quantum phase estimation \cite{babbush2025grand}, with cheaper approximation of the filter-function integral as an important target for further optimization.

\section{The single-ancilla Lindblad simulation algorithm}
\label{sec:algorithm}

Preparing states via Lindblad simulation requires solving two problems: i) constructing a quantum algorithm that simulates the Lindblad evolution
\begin{align}
\frac{d\rho}{dt} = \mathcal{L}[\rho] ,
\end{align}
to the desired accuracy \cite{cleve1612efficient} and ii) engineering a suitable Lindbladian, 
\begin{align}
 \mathcal{L}[\rho] &= \underbrace{-i [H, \rho]}_{=\Lind_H} + \underbrace{\sum_{k} \Big(
K_k \rho K_k^\dagger
- \frac{1}{2} K_k^\dagger K_k \rho
- \frac{1}{2} \rho K_k^\dagger K_k
\Big)}_{=\Lind_K}, \label{eqn:lindblad_def}
\end{align}
where the coherent evolution operator, $H$, and the transition operators, $\{K_k\}$, are chosen so that the resulting dynamics converge to the desired fixed point, typically a finite-temperature Gibbs state \cite{chen2023quantum} or the ground state of the system \cite{ding2024single}.

\subsection{Dilation-based simulation}

For the Lindblad simulation, there are two main approaches: The theoretically-optimal scaling in simulation time and error is achieved through a linear-combination-of-unitaries construction, but it has relatively high ancilla overhead \cite{cleve1612efficient,leng2026accelerating}. A potentially more near-term approach is given by {\it dilation-based algorithms} \cite{cleve1612efficient}, which have worse scalings but require as little as a single ancilla to implement the Lindblad evolution \cite{ding2024single}. Randomized compilation can recover logarithmic precision scaling with logarithmically many ancillas \cite{kato2026}, at the cost of quadratic time scaling and a restriction to observable estimation.
Here we focus on the single-ancilla approach, which implements a single jump operator $K$ per simulation step, drawn from a classical probability distribution. This works by implementing unitary evolution on an extended register and then tracing out the ancilla qubit, using the dilated evolution operator
\begin{align}
     \tilde{K} &= \begin{pmatrix}
            0 & K^\dagger \\
            K & 0 
        \end{pmatrix} .
\end{align}
Each step of the dissipative evolution of duration $\tau$ is implemented via the channel \cite{ding2024single}
\begin{align}
    \Phi(\tau)[\rho] := \operatorname{Tr}_a \!\left[
e^{-i \tilde{K} \sqrt{\tau}}
\left( \lvert 0 \rangle \langle 0 \rvert \otimes \rho \right)
e^{i \tilde{K} \sqrt{\tau}}
\right] .
\end{align} 
It can be shown that this implements the underlying Lindblad evolution to second order in $\tau$ (\cite[Lemma 3 and 9]{ding2024single}):
\begin{align}
    \tnorm{e^{\Lind \tau}[\rho]- e^{\Lind_H \tau}\Phi(\tau)[\rho]} = \mathcal{O}(\tau^2),
\end{align}
where $\tnorm{\, \cdot \,}$ denotes the Schatten-1 or trace norm, $\tnorm{A}=\Tr\sqrt{A^\dagger A}$. More transition operators can be applied by using additional ancillas, but usually the desired evolution can be approximated by stochastically sampling the transition operators in each step. The Lindblad simulation for arbitrary times can then be performed by chaining together these time-evolution-trace-out steps with a suitably small $\tau$, similarly to Trotterization.

\subsection{Transition operator}

\begin{figure}
    \centering
    \includegraphics[width=0.6\linewidth]{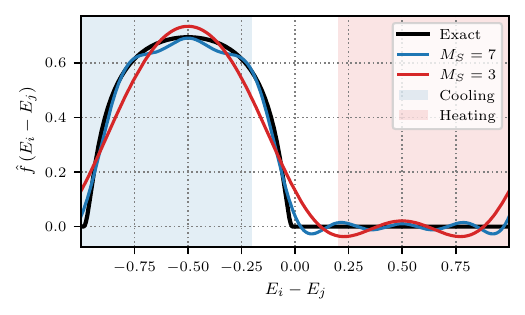}
    \caption{The energy-domain filter function $\hat{f}$ as well as its Dirichlet kernel approximations for different numbers of sample points $M_S$. The colored areas begin at the smallest possible energy transitions to and from the ground state, $\Delta$. The blue area indicates transitions that dissipate energy (cooling) and the red area indicates transitions that increase the energy (heating). }
    \label{fig:filter_fn_energy_domain}
\end{figure}

The dilation-based framework can be used for both ground state and Gibbs state preparation \cite{ding2025end}. The key design step is taking a generic jump operator $A$, such as a Pauli operator or hopping term, and convolving it with a {\it filter function} $f(s)$ that appropriately suppresses certain energy transitions:
\begin{align}
    K :=& \sum_{i,j \in [N]} \hat{f}(E_i - E_j)
\, \lvert \psi_i \rangle \langle \psi_i \rvert
A
\lvert \psi_j \rangle \langle \psi_j \rvert  \label{eqn:K_energy} \\
=& \int_{-\infty}^{\infty} f(s)\, A(s)\, ds . \label{eqn:K_integral}
\end{align}
Here, $ \{ \lvert \psi_i \rangle\}$ are the eigenstates of the system Hamiltonian, and $A(s) = e^{iHs}A e^{-iHs}$ is the time-evolved jump operator. The energy-domain filter function $\hat{f}(\omega)$ determines the allowed energy transitions. By choosing $\hat{f}(\omega) = 0$ for $\omega>0$ we guarantee that $ \lvert \psi_0 \rangle$ is a fixed point of $e^{i \sqrt{\tau}\tilde{K}}$, i.e. 
\begin{align}
K\lvert \psi_0 \rangle =  0,\label{eqn:K_fixed_point}
\end{align}
and therefore of the Lindblad dynamics.\footnote{This does not guarantee ergodicity, which has to be ensured through the choice of transition operator.} 
In practice we choose a bump function as the energy-domain filter function 
\begin{align}
    \hat{f}(\omega) = \mathbbm{1}_{[-S_\omega, 0]}  \exp\left\{ -\frac{a}{1-(\frac{2}{S_\omega} \omega + 1)^2} \right\} , \label{eqn:filter_def}
\end{align}
which has support only on $\omega \in [-S_\omega, 0]$, see \cref{fig:filter_fn_energy_domain}. In order to allow all possible energy-decreasing transitions we choose $S_\omega$ to cover the spectral width of $H$, i.e. we need $S_\omega \geq 2\norm{H}$ for a symmetric spectrum. The steepness parameter, $a$, sets the slope of the flanks of the bump function, which approaches the indicator function as $a \to 0$.

\subsection{Approximate energy filter}
\label{sec:approximate_energy_filter}

In practice, one approximates the integral over the entire real line in \cref{eqn:K_integral} by finite quadrature: 
\begin{align}
    K\simeq K_S:= & \sum_{l=-M_S}^{M_S} f(s_l) e^{iH s_l}A e^{-iH s_l} w_l ,
    \label{eqn:K_S_def}
\end{align}
with $s_l = l \tau_S$ and $w_l = \tau_S$ except for the boundary points, $w_{\pm M_S} = \tau_S/2$. One then has to decide how to determine the quadrature parameters $M_S$ and $\tau_S$ in order to ensure convergence to the ground state. Since the bump function has finite support, the Nyquist theorem tells us that the time step $\tau_S = \pi /S_\omega$ is optimal.

For this choice, we show in \cref{sec:quad_error} that for the bump function the quadrature error decreases exponentially with the number of integration points $M_S$:
\begin{align}
    \varepsilon_{M_S} \propto e^{-\sqrt{ a M_S}} . \label{eqn:eps_quad_scaling}
\end{align}
To ensure that the filter discriminates between the smallest energy increase and the smallest energy decrease, the flank of the filter function must be steep enough to resolve the spectral gap $\Delta$, as illustrated in \cref{fig:filter_fn_energy_domain}. To allow for a transition rate of at least $p_{\text{min}}$ for the lowest-energy transition, we require
\begin{align}
    a &= \ln(1/p_{\text{min}}) \left( 1 -(1 - 2\Delta/S_\omega )^2 \right) \label{eqn:optimal_a} \\
    &\stackrel{\Delta \ll S_\omega}{\simeq}  \Delta/S_\omega . \label{eqn:a_scaling}
\end{align}
Combining \cref{eqn:eps_quad_scaling} and \cref{eqn:a_scaling} recovers the well-known relationship between measurement time and spectral resolution: To resolve a spectral gap, $\Delta$, we need to sample a time interval of size $M_S \tau_S \propto 1/\Delta$.

To understand the effect of quadrature on the energy domain further, we insert $f(s) = \int_{-\infty}^{\infty} \frac{d\omega}{2\pi} \, e^{-i\omega s} \hat{f}(\omega)$ into \cref{eqn:K_S_def}:
\begin{align}
    K_S=& \sum_{l=-M_S}^{M_S} f(s_l) e^{iH s_l}A e^{-iH s_l} w_l\\
    =& \sum_{l=-M_S}^{M_S} \int_{-\infty}^{\infty} \frac{d\omega}{2\pi} \, e^{-i\omega s_l} \hat{f}(\omega)\sum_{i,j \in [N]} e^{iE_i s_l} \lvert \psi_i \rangle \langle \psi_i \rvert A \lvert \psi_j \rangle \langle \psi_j \rvert  e^{-iE_j s_l} w_l \\
    =& \sum_{i,j \in [N]} \lvert \psi_i \rangle \langle \psi_i \rvert A \lvert \psi_j \rangle \langle \psi_j \rvert  \int_{-\infty}^{\infty} \frac{d\omega}{2\pi} \,  \hat{f}(\omega) \underbrace{\sum_{l=-M_S}^{M_S} w_l \,  e^{-i(\omega-(E_i-E_j)) s_l}}_{= \tau_S D_{M_S}((\omega-(E_i-E_j))\tau_S)}\\
    =& \sum_{i,j \in [N]} \lvert \psi_i \rangle \langle \psi_i \rvert A \lvert \psi_j \rangle \langle \psi_j \rvert  \underbrace{\int_{-\infty}^{\infty} \frac{d\omega}{2\pi} \,  \hat{f}(\omega)\tau_S D_{M_S}((\omega-(E_i-E_j))\tau_S)}_{\hat{f}_\text{eff}(E_i-E_j)} \\
    =& \sum_{i,j \in [N]} \lvert \psi_i \rangle \langle \psi_i \rvert A \lvert \psi_j \rangle \langle \psi_j \rvert   \, \hat{f}_\text{eff}(E_i-E_j),\label{eqn:K_S_dirichlet}
\end{align}
where we identified the truncated sum with the shifted Dirichlet kernel
\begin{align}
    D_{M_S}(x)
:= \sum_{k=-M_S}^{M_S} e^{ikx} - \frac{1}{2} e^{iM_Sx} - \frac{1}{2} e^{-iM_Sx}
= \frac{\sin(M_S x)}{\tan(x/2)} , \label{eqn:dirichlet}
\end{align}
and defined the effective filter function for $\Delta E_{ij}=E_i-E_j$
\begin{align}
    \hat{f}_\text{eff}(\Delta E_{ij}) = \tau_S \int_{-\infty}^{\infty} \frac{d\omega}{2\pi} \,  \hat{f}(\omega)\,  D_{M_S}\left((\omega-\Delta E_{ij})\tau_S\right) . \label{eqn:def_f_eff}
\end{align}
In the limit of large $n$ the Dirichlet kernel converges to the Dirac comb and the effective filter function $\hat{f}_\text{eff}$ converges to $\hat{f}$. This is illustrated in \cref{fig:filter_fn_energy_domain}. In this limit we also recover the exact expression for $K$, \cref{eqn:K_energy}. 

This picture gives us a more general way to derive the relationship between the spectral gap and the quadrature parameters in \cref{eqn:a_scaling} from the energy domain picture, \cref{fig:filter_fn_energy_domain}: For an arbitrary filter function that approximates an indicator function, $\mathbb{1}_{[-S_\omega,0]}$, (the limit $a \ll 1$ in the case of a bump function) the transition probability for energy increasing transitions comes entirely from the Dirichlet kernel in \cref{eqn:K_S_dirichlet}. We demand that the effective filter function in \cref{eqn:def_f_eff} strongly suppresses transitions from the ground state to the first excited state, i.e. those with $\Delta = E_1-E_0>0$:
\begin{align}
    \hat{f}_\text{eff}(\Delta) \simeq & \tau_S \int_{-S_\omega}^{0} \frac{d\omega}{2\pi} \, \mathbb{1}_{[-S_\omega,0]}   D_{M_S}\left((\omega-\Delta)\tau_S\right) \\
    =& \int_{-(S_\omega+\Delta)\tau_S}^{-\Delta \tau_S} \frac{du}{2\pi} \, D_{M_S}(u) \\
    =& \int_{-(S_\omega+\Delta)\tau_S}^{-\Delta \tau_S} \frac{du}{2\pi} \, \frac{\sin(M_Su)}{\tan(u/2)} .
\end{align}
The $L_1$-norm of the Dirichlet filter is concentrated in its central peak, which is bounded by the zeros of the $\sin$-function at $M_Su=\pm \pi$. If the upper bound of the integral does not include this main flank then the transition is strongly suppressed, which is true for $M_S \Delta \tau_S = \pi$. Inserting the Nyquist condition $\tau_S= \pi/S_\omega$ recovers
\begin{align}
    M_S \geq \frac{S_\omega}{\Delta} , \label{eqn:dirichlet_scaling}
\end{align}
which again illustrates that in order to achieve a spectral resolution $\Delta / S_\omega$ we need an inversely proportional number of sampling points. The effect of the finite quadrature can therefore be understood entirely through convolution of the filter function with the Dirichlet kernel in the energy domain, a fact we will later use when choosing the quadrature parameters. 

\subsection{Error analysis}

To go from a given Lindbladian, \cref{eqn:lindblad_def}, to a circuit that can be implemented on a quantum computer requires a series of approximations. The exact definitions of the steps can be found in \cref{sec:errors}, but at a high level the approximations are:
\begin{enumerate}
\item Splitting the Lindbladian into its coherent and dissipative parts.\footnote{This can be circumvented by integrating the coherent term into the dilation as done in \cite{hahn2026efficient, ding2025end}, but this splitting error is subdominant and would therefore not change the overall analysis.}
\item Using the dilation technique described above to approximate the dissipative part.
\item Approximating the integral via quadrature. 
\item Trotterizing the time sum.
\item Bounding the Trotter error of the time evolution with $A$ and $H$.
\end{enumerate}

Because completely positive trace preserving (CPTP) maps are contractive under the trace norm, it follows that we can bound the error of the evolution for time $T=M \tau$ by only considering the error in a single step and multiplying it by $M$. The different error sources can then be isolated via the triangle inequality:
\begin{align}
\begin{split}
    \norm{\rho(\tau) - \bar{\rho}_{1,\text{trott}}}_1 \leq & \underbrace{\norm{e^{(\mathcal{L}_H + \mathcal{L}_K) \tau}\rho(0) - e^{\mathcal{L}_H \tau} e^{\mathcal{L}_K \tau} \rho(0)}_1}_{\text{Lindblad splitting error}} 
    + \underbrace{\norm{e^{\mathcal{L}_H \tau} e^{\mathcal{L}_K \tau} \rho(0) - \tilde{\rho}_1}_1}_{\text{Dilation error}} \\
    & + \underbrace{\norm{\tilde{\rho}_1-\tilde{\rho}_{1,s}}_1}_{\text{Integral quadrature error}} 
    + \underbrace{\norm{\tilde{\rho}_{1,s}-\bar{\rho}_{1}}_1}_{\text{Time sum Trotterization error}} 
    +\underbrace{\tnorm{\bar{\rho}_{1}-\bar{\rho}_{1,\text{trott}}}}_{\text{Trotter evolution error}}.\label{eqn:error_overview}
\end{split}
\end{align}

The authors of \cite{ding2024single} derive the asymptotic scaling of the first four errors assuming no Trotter error for $e^{-itH}$ and $e^{-itA}$. We extend the analysis here by deriving the full error bounds:

\begin{proposition}
\label{prop:total_error}
    Let $\bar{\rho}_{M,\text{trott}}$ be the state after $M$ steps of the fully compiled algorithm with step duration $\tau$, total time $T = M\tau$, quadrature spacing $\tau_S$ with $2M_S+1$ points, and Trotter step sizes $\tau_{T,1} = \tau/M_{T,1}$ and $\tau_{T,2} = \tau_S/M_{T,2}$ for the coherent and dissipative $H$ evolutions, respectively. Then, for any initial state $\rho(0)$,
    \begin{align}
    \begin{split}
        \tnorm{\rho(T) - \bar{\rho}_{M,\text{trott}}} \leq T \Big[ & \tau \left( C_S + C_D + C_{T,1} \right) + \tau^2 C_{T,2} + C_Q \\
        & + \Big( 2 \tau_{T,1}^2 + 4 (2M_S+1) \frac{\tau_S}{\tau} \tau_{T,2}^2 \Big) C_H + \tau_S^2 \, C_{\tilde{A}} \Big] ,
    \end{split} \label{eqn:total_error}
    \end{align}
    with the constants
    \begin{align}
        C_S &:= 4 \norm{H} \norm{A}^2 F_\text{tot}^2 , \\
        C_D &:= \frac{8}{3} \norm{A}^4 F_\text{tot}^4 , \\
        C_Q &:= 4 \norm{A}^2 F_\text{tot} F(M_S) , \\
        \begin{split}
    C_{T,1} &:=  \sum_{k=1}^{K}  \left[ \frac{1}{24} \norm{\ad^3_{V_{k}}(v_{k})} + \frac{1}{96} \norm{\ad^3_{v_{k}}( V_{k})} \right. \\
        & \quad \left. + \frac{1}{24}  \Big\| \ad^2_{V_k} v_k
        -  \frac{1}{2}\ad^2_{v_k} V_k \Big\|
        \Big( 11 \norm{V_{k+1}} + 3 \norm{V_k} + 3 \norm{v_k} + 8\sum_{k'=k+1}^K \norm{v_{k'}} \Big)\right] ,
    \end{split}\\
        C_{T,2} &:= \sum_{k=1}^{K}  \frac{1}{144} \left( \norm{\ad^2_{V_{k}} (v_{k})} +  \frac{1}{2} \norm{\ad^2_{v_{k}}(V_{k})} \right)^2  ,\\
        C_H &:= \frac{1}{12} \sum_{\gamma_1=1}^{\Gamma} \left\| \left[ \sum_{\gamma_3=\gamma_1+1}^{\Gamma} H_{\gamma_3}, \left[ \sum_{\gamma_2=\gamma_1+1}^{\Gamma} H_{\gamma_2}, H_{\gamma_1} \right] \right] \right\| + \frac{1}{24} \sum_{\gamma_1=1}^{\Gamma} \left\| \left[ H_{\gamma_1}, \left[ H_{\gamma_1}, \sum_{\gamma_2=\gamma_1+1}^{\Gamma} H_{\gamma_2} \right] \right] \right\| , \label{eqn:general_trotter_constant} \\
        C_{\tilde{A}} &:=  \norm{A}^2 \sum_{l=-M_S}^{M_S} \left| \operatorname{Im} f(s_l) \right| \left| \operatorname{Re} f(s_l) \right| .
    \end{align}
    Here $F(M_S) \propto \exp(-\sqrt{aM_S})$ is the filter tail mass bounded in \cref{eqn:filter_tail}, $F_\text{tot} = \sum_{l=-\infty}^{\infty} |f(s_l)|\tau_S$, $v_k$, $V_k$ and $K = 2M_S+1$ are defined in \cref{prop:quad_trott_error}, and $C_H$ is the second-order Trotter constant of the system Hamiltonian, $H=:\sum_{\gamma=1}^{\Gamma} H_\gamma$.
\end{proposition}
The proof of \cref{prop:total_error}, with details of the notation used, can be found in \cref{sec:errors}. 
We further show in \cref{sec:fixed_point_error} that the splitting and dilation errors do not necessarily contribute to the error between the fixed points of the exact and approximate evolution, so for the task of state preparation these can be ignored when choosing the parameters.

From this, we can recover the known scalings of the algorithm parameters that are needed for target trace-norm precision $\varepsilon_\text{Tr}$: In the quadrature error analysis we set the integral step size to the Nyquist frequency, 
\begin{align}
    \tau_S = \frac{\pi}{S_\omega} ,
\end{align}
which leads the number of sample points to scale as
\begin{align}
    M_S \propto \frac{\norm{H}}{\Delta} \text{polylog}(T_\text{mix}/\varepsilon_\text{Tr}) .
\end{align}
Note that this depends only polylogarithmically on the mixing time, so $M_S$ can approximately be chosen irrespectively of the mixing time, provided that the other errors are controlled by the overall time step $\tau$, which scales as 
\begin{align}
    \tau = \tilde{\mathcal{O}} \left( \frac{\varepsilon_\text{Tr}}{T_\text{mix}(\norm{K}^4 +\norm{K}^2 \norm{H})} \right) .
\end{align}
The system in question only enters here in the second term under the denominator, which comes from the Lindblad splitting error. This is important when we want to guarantee that the simulation stays close to the ideal evolution of the Lindbladian. But, as we discuss in \cref{sec:fixed_point_error}, this is not necessary for the ground state to be a fixed point of the dynamics. For systems that we know to converge to the ground state from simulations we omit this dependence, leading to system independent scaling up to polylogarithmic factors
\begin{align}
    \tau = \tilde{\mathcal{O}} \left( \frac{\varepsilon_\text{Tr}}{T_\text{mix} \norm{K}^4 } \right) . \label{eqn:fixed_point_tau_scaling}
\end{align}

\section{Simulation}
\label{sec:simulation}

We implement two types of simulations for the Lindblad algorithm: a Cirq-based \cite{Cirq_Developers_2026} density matrix simulation and a Qiskit-based \cite{javadi2024quantum} state vector simulation. The former is deterministic, but requires memory to store the entire density matrix, limiting it to very small system sizes. The state-vector simulator, on the other hand, needs to sample multiple shots to estimate the average trajectory, but can simulate somewhat larger systems, although in practice runtime becomes a concern before memory limits are reached. There are more sophisticated techniques based on matrix product states as proposed in \cite{zhan2026rapid}, but these are not necessary for the systems in consideration here. 

\subsection{Hubbard model}

We simulate the fermionic Hubbard model,
\begin{align}
H
= \underbrace{-t_\text{hop} \sum_{i,\sigma} \left( \hat{c}^\dagger_{i,\sigma}\,\hat{c}_{i+1,\sigma}
+ \hat{c}^\dagger_{i+1,\sigma}\,\hat{c}_{i,\sigma} \right)}_{H_h}
+ \underbrace{U \sum_i \hat{n}_{i\uparrow}\,\hat{n}_{i\downarrow}, }_{H_I}\label{eqn:hubabrd_def}
\end{align}
where $t_\text{hop}$ is the hopping parameter and $U$ is the interaction strength. In the following we focus on the intermediate-coupling regime with $t_\text{hop}=1$ and $U=4$, and consider a one-dimensional lattice with an even number of sites and open boundary conditions\footnote{For the implementation and the parameter choices the geometry of the system makes little difference, but we consider 1D systems here because they allow multiple data points before we reach the limit of our simulations at around 8 sites. For actual quantum advantage one has to consider more complex geometries, but the $S_1$ tiling that we consider below can be adapted for almost any lattice with minimal modification.}. The number of qubits required for the Lindblad simulation of such a system (applying a Jordan-Wigner encoding) is one qubit for each spin-orbital plus one qubit for the ancilla, leading to a total of $2 n_{\text{sites}} +1$ qubits. In the following all times and energies are understood to be in units of $1/t_\text{hop}$ and $t_\text{hop}$, respectively. 

As the initial state we choose the half-filling N\'eel state with $S_z = 0$, $\ket{\uparrow \downarrow \uparrow \downarrow...\downarrow}$, which can be prepared trivially. It is also known that, for systems with an even number of sites, the global ground state is included in this sector. To verify that our results hold more generally, we initialized the simulation with all the excited eigenstates of the 4-site system in this spin-particle-number sector, which have zero overlap with the ground state, and found a similar convergence behavior for all of them. 

Crucial for this is the choice of mixing operator, $A$: In practical applications the global ground state across all symmetry sectors might not be relevant, but instead the lowest-lying energy state within a certain spin and particle-number sector is the target quantity. For example, one is often interested in the Hubbard model at a particular doping away from half-filling. 
We therefore need to use an operator that commutes both with the total particle number operator and the total spin operator, and for easy implementation we require it to be local. The most general single-edge hopping operator with mixer strength $m\in \mathbb{C}$ is
\begin{align}
    A = m \, c^\dagger_{i,\sigma} c_{j,\sigma} + m^* \,c^\dagger_{j,\sigma} c_{i,\sigma} , \label{eqn:A_def}
\end{align}
although we find that only the imaginary part of $m$ contributes towards faster convergence, so we set Re$(m)=0$ in the following. We only apply the hopping mixer on one spin at a time, as it otherwise loses ergodicity, see \cref{fig:mixer_comparison}.

\begin{figure}
    \centering
    \includegraphics[width=0.7\linewidth]{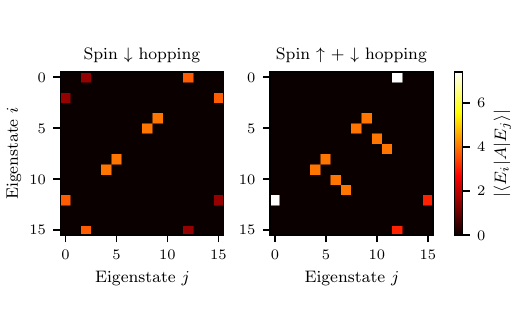}
    \caption{Comparison of jump operators, $A$, in the eigenbasis of the 2-site Hubbard Hamiltonian with $U=4$ and open boundary conditions. The left operator acts only on the spin-down section, whereas the right operator acts on a linear combination of both. The axis labels indicate the eigenstates in order from lowest to highest energy, such that the top-right triangular matrix contains the desired energy transitions. Note that the two-spin mixer on the right is no longer ergodic in the $S_z=0$, 2-particle sector since it is missing the top-left transition.}
    \label{fig:mixer_comparison}
\end{figure}

\subsection{Algorithmic error resilience}

Unlike unitary evolution, Lindblad evolution is contractive, meaning that it has a certain resilience to noise: If the noise is weak enough it can be ``offset'' by the Lindblad dynamics, which then converges to a new fixed point that may still be close to the target state. Such behavior has, for example, been observed in spin systems \cite{hahn2026efficient}.

In practice, this approach suffers whenever we are only interested in a subspace of the total Hilbert space of the system. This requires that the jump operators, and thus the Lindbladian, are only ergodic over this subspace. However, this means that noise that does not respect the computational subspace causes the system state to leak out, and the Lindblad dynamics do not allow us to return to the original space. This is illustrated in \cref{fig:noise}, which shows that even low levels of circuit-level noise cause the fixed point of the evolution to drift over time.

\begin{figure}
    \centering
    \includegraphics{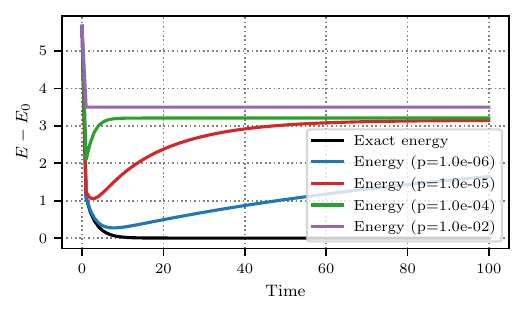}
    \caption{Energy error relative to the ground-state energy as a function of time, for the 2-site Hubbard model with $U=4$, with different levels of amplitude damping noise applied. The system is initialized in a $S_z=0$, two-particle state, but drifts into other particle sectors due to the noise, resulting in a different fixed point.}
    \label{fig:noise}
\end{figure}

Unfortunately, the case where we want to restrict the evolution to a particular particle-number sector is typical in fermionic systems such as the Hubbard model or quantum chemistry\footnote{This is obvious in second quantization, where the Hilbert space includes the multi-particle Fock space, but also affects first quantization \cite{georges2025quantum}, where most algorithms rely on the initial state preparation to return an anti-symmetrized wave function, which would no longer be guaranteed.}. The noisy evolution will leak out of such a subspace. If the subspace of interest contains the global energy minimum and the subsequent algorithm is resilient to a finite overlap then we can mitigate this issue by introducing globally ergodic mixers, although these might lead to longer mixing times. However, to answer questions about superconductivity, as one example, one may wish to study systems doped away from half-filling, such that it is not sufficient to find the global ground state across all symmetry sectors.

\subsection{Precision-speed trade-off}

\begin{figure}
    \centering
    \includegraphics{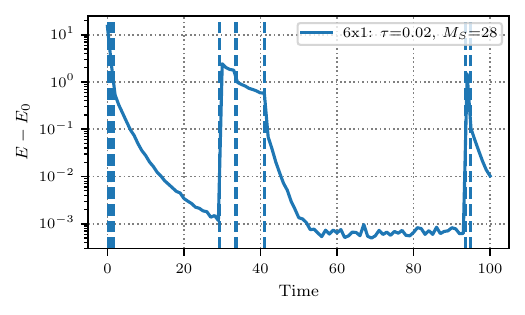}
    \caption{Example energy trajectory in a state vector simulation for a 6-site one-dimensional Hubbard model with $U=4$ and open boundary conditions. Vertical dashed lines indicate jumps, i.e. when the ancilla state was measured in the $\ket{1}$-state. }
    \label{fig:trajectory_example}
\end{figure}

A typical state vector trajectory is shown in \cref{fig:trajectory_example} with jumps, i.e. cases where we measure $\ket{1}$ on the ancilla, indicated by dashed vertical lines. The output state of the quantum simulation is not deterministic, because the results depend on the outcomes of the mid-circuit measurements. Jumps are common far from the ground state and lead to fast energy dissipation, but become rarer as the system approaches the ground state. This is because the jump probability is proportional to $\norm{K \ket{\psi}}^2$, which vanishes for the ground state, $\ket{\psi_0}$. 
However, for the approximate filter integral, \cref{eqn:K_S_def}, there is a finite probability of jumps that increase the energy. In fact, for the ground state these are the only jumps that can happen. In \cref{fig:trajectory_example} these show up as spikes in the energy. This can lead to temporary high-energy states after a trace-out step even after the density matrix has already reached the fixed point, because the trace norm bounds only restrict the expectation values over the trace operation.

\begin{figure}
    \centering
\includegraphics{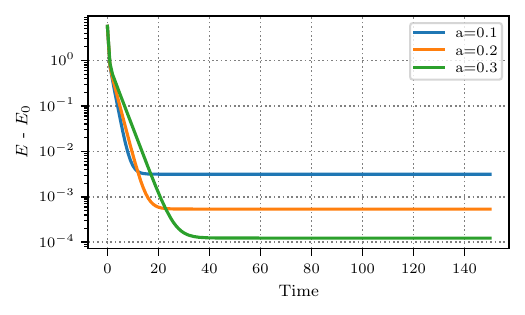}
    \caption{Energy convergence in the density matrix simulation of the 2-site Hubbard model with $U=4$ for different values of the steepness parameter $a$. Smaller values of $a$ speed up the evolution as they enable more dissipation, but they also increase leakage into higher-energy states.}
    \label{fig:speed_time_tradeoff}
\end{figure}

We find a trade-off between convergence speed and final precision, as shown in \cref{fig:speed_time_tradeoff}: By increasing the steepness parameter $a$ we can suppress energy-increasing transitions exponentially, but at the cost of reducing the transition rate for small energy transitions. This can be seen in \cref{fig:filter_fn_energy_domain}, where the steepness parameter determines the amplitude of the flank of $\hat{f}$ at the spectral gap. In practice, we find that choosing $a \ll 1$, i.e. the regime where $\hat{f}$ approaches a rectangular filter, ensures fast convergence, if the approximation error is fully captured in the Dirichlet kernel, which can be controlled by choosing $M_S$ as discussed in \cref{sec:approximate_energy_filter}.

\subsection{Optimal parameter regimes}

For a given error budget we have to choose the number of approximation points, $M_S$, to approximate the filter integral, \cref{eqn:K_S_def}. Higher values of $M_S$ allow lower fixed point error, but require longer Hamiltonian simulation times, since we need to time evolve $e^{-iH t_\text{step}}$, with $t_\text{step}=2(2M_S+1)\tau_S$, in every time step. 

In the {\it rigorous regime}, we use that the quadrature error is suppressed as $\varepsilon_{M_S} \propto \exp(-\sqrt{a M_S})$, such that choosing $a$ according to \cref{eqn:a_scaling} and selecting $M_S$ accordingly we can exponentially suppress the quadrature error while ensuring that we can 
discriminate between all relevant energy transitions.

In the {\it empirical scaling regime} we choose $a \ll 1$, since numerically we see the fastest convergence there. In this limit the filter approximates an indicator function and we choose $M_S$ according to the heuristic \cref{eqn:dirichlet_scaling}, which has the same asymptotic scaling with the system parameters, $M_S \propto \norm{H}/\Delta$, but we determine the prefactor empirically from the simulation results, which is much lower than in the rigorous regime.

However, this still requires $M_S\geq 10$, even for the two-site model. We can do better if we know the exact gap of the system: Because the effective filter function \cref{eqn:def_f_eff} is easy to calculate classically at any system size, we can search for an $M_S$ that ensures that the first node of the effective filter function is placed exactly on the spectral gap, so that the transition out of the ground state to the first excited state is strongly suppressed. This {\it empirical scaling with known gap} has the same asymptotic behavior as the general empirical and rigorous scalings, but with an even lower prefactor. This allows fixed-point energy errors as low as $10^{-2}$ with as few as 9 sample points, $M_S=4$, see \cref{fig:optimal_M_S_energy_domain} for an illustration. This does however require knowing the gap of the system in advance, which is usually not the case and would in fact make the calculation redundant if the goal of the algorithm is to estimate low-lying energies. We include these results nevertheless, because preparation of states with known energies could be a useful subroutine for other algorithms.

\begin{figure}[h]
    \centering
    \includegraphics{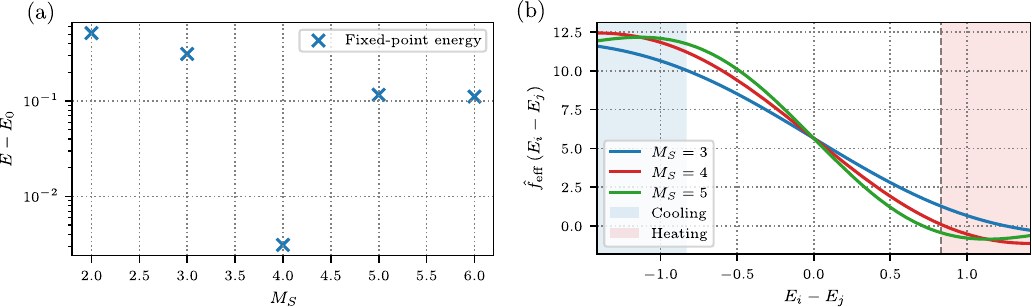}
    \caption{We plot the dependence of the fixed-point energy on the number of sampling points, $M_S$, for the 2-site Hubbard model with $\tau=0.125$. (a) $M_S=4$ gives the optimal energy error, even compared to better approximations with larger $M_S$. (b) In the energy domain we see that the node of $\hat{f}_\text{eff}$ for $M_S=4$ lies exactly at the spectral gap, as indicated by the dashed line. This causes the transition rate from the ground state to the first excited state to be nearly zero.}
    \label{fig:optimal_M_S_energy_domain}
\end{figure}

For the Hubbard model one often considers an extensive error budget \cite{Kivlichan2020improvedfault, campbell2022early}, $\varepsilon = 0.005\cdot n_\text{sites}$, where $n_\text{sites}$ is the number of lattice sites. In particular, this extensive error budget is appropriate when we aim to estimate the ground-state energy \emph{per site} to a fixed accuracy, as is common. We therefore do not target the exact ground state, but instead the low-energy subspace with $E \leq E_0 + \varepsilon$, and allow for a trace-norm error of $\varepsilon_\text{Tr}=0.005 \cdot n_\text{sites} / \norm{H}$.

The empirical scaling and the ``empirical scaling with known gap'' for the extensive energy error budget is determined in the following way: First, we find minimal parameters that reach this energy precision by scanning the parameter landscape for a 2-site system. Once found, we extrapolate to larger systems using the known asymptotic scaling laws, but using the prefactors found in the simulation. We then verify that the scaled parameters satisfy the desired precision numerically by performing state vector simulations for up to 8 sites, where the exact energy and spectral gap are known. 
For larger systems we extrapolate the spectral norm and gap from the smaller systems, which then enter into the empirical scaling laws as parameters, see \cref{sec:spectral_scaling}. For the rigorous scaling estimate we use the known bound for the spectral norm \cite{campbell2022early}
\begin{align}
    \norm{H} \leq \norm{H}_\text{max} := t_{\text{hop}} \tnorm{R} + U n_\text{sites} , \label{eqn:H_spectral_bound}
\end{align}
where $R$ is the adjacency matrix of the lattice graph. 
For the empirical scalings we can diagonalize the Hamiltonian for systems with up to 8 sites and extrapolate from there with power law fits as described in \cref{sec:spectral_scaling}. With these estimates for the spectral norm and gap we then use the verified parameter scaling laws to estimate the required parameters for classically-non-trivial systems.

For large systems we reach a crossover point, where the target extensive energy error becomes larger than the spectral gap, $\varepsilon = 0.005\cdot n_\text{sites} > \Delta$. In this case there is no need to resolve these fine energy structures, and we use the effective spectral gap
\begin{align}
    \Delta_\text{eff}:=\max(0.005\cdot n_\text{sites},  \, \Delta) ,
\end{align}
instead. We compare the resulting parameter regimes in \cref{tab:scaling_overview}.

\begin{table}[tb]
\centering
\renewcommand{\arraystretch}{1.6}
\setlength{\tabcolsep}{8pt}
\begin{tabular}{@{}lccc@{}}
\toprule
 & Rigorous bounds & Empirical (general) & Empirical (known gap) \\
\midrule
$S_\omega$
  & $2\norm{H}_\text{max}$
  & $2.2\norm{H}$
  & $2\norm{H}$ \\
$a$
  & $\ln(1/0.05)\bigl[1-(1-2\Delta_\text{eff}/S_\omega)^2\bigr]$
  & $10^{-4}/n_\text{sites}$
  & $10^{-6}$ \\
$M_S$
  & $\propto a^{-1}\log^2(1/\varepsilon_{M_S})$ \ (\ref{eqn:quadrature_error})
  & $1.23\,S_\omega/\Delta_\text{eff} - 0.5$
  & $0.55\,(S_\omega/\Delta_\text{eff})^{1.03}$ \\
$M = T/\tau$
  & $\propto T_\text{mix}^2\, n_\text{sites}/\varepsilon_\tau$ \ (\ref{eqn:total_error})
  & $\propto T_\text{mix}^2/\varepsilon_\tau$ \ (\ref{eqn:fixed_point_tau_scaling})
  & $\propto T_\text{mix}^2/\varepsilon_\tau$ \ (\ref{eqn:fixed_point_tau_scaling}) \\
$M_{T,1},\, M_{T,2}$
  & $\propto T_\text{mix}^{3/2}\, n_\text{sites}\,\varepsilon_{H,\text{trott}}^{-1/2}$ \ (\ref{eqn:H_trott_err_scaling})
  & $2$
  & $2$ \\
$m$
  & $4i$
  & $6i$
  & $8i$ \\
\bottomrule
\end{tabular}
\caption{Parameter choices in the three regimes: rigorous bounds, empirical fits, and empirical fits with known gap. The parameters are the spectral width $S_\omega$, the steepness parameter $a$, the number of quadrature points $M_S$, the number of Lindblad steps $M$, the numbers of Trotter steps $M_{T,1}$ and $M_{T,2}$ for $e^{-i\tau H}$ and $e^{-i\tau_S H}$, respectively, and the mixer strength $m$. Equation numbers indicate where each scaling is derived.}
\label{tab:scaling_overview}
\end{table}

We validate these scalings for systems with 2, 4, 6 and 8 sites. The energy trajectories of the state vector simulations, averaged over the multiple independent runs, are shown in \cref{fig:parameter_scaling_simulations}.

\begin{figure}[h]
    \centering
    \begin{minipage}{0.49\textwidth}
        \centering 
        \includegraphics[width=\linewidth]{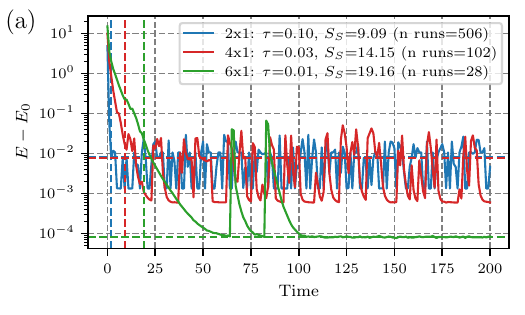}
    \end{minipage}\hfill
    \begin{minipage}{0.49\textwidth}
        \centering 
        \includegraphics[width=\linewidth]{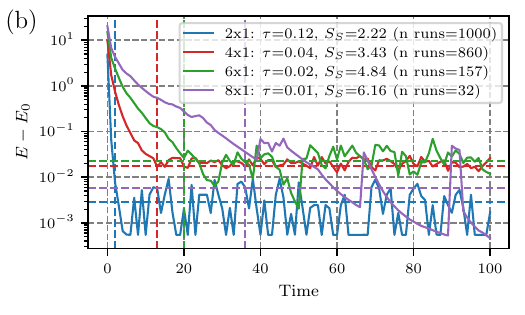}
    \end{minipage}
    \caption{Averaged energy trajectories for one-dimensional Hubbard models of varying sizes. In the legends, $\tau$ is the Lindblad step size, and $S_S$ is the Hamiltonian evolution time required for the filter integral. The number of repetitions over which the average is taken is $n {\text{ runs}}$, shown in the legend for each result. (a) shows simulations with parameters from the empirical scaling, while (b) shows the simulations using the empirical scaling with known gap. Horizontal dashed lines indicate the average energies in the long-time limit, and vertical dashed lines indicate the times at which the average falls within the target energy precision, $\varepsilon = 0.005\cdot n_\text{sites}$.}
    \label{fig:parameter_scaling_simulations}
\end{figure}

\section{Resource estimation}
\label{sec:costing}

We now turn to estimating the fault-tolerant resources required for simulating Hubbard models, specifically to prepare a state in the low-energy subspace. We primarily focus on resource estimates for one-dimensional lattices, but also argue that the cost \emph{per unit simulation time} will not be significantly different for two-dimensional lattices.

The two main primitives required for the Lindblad evolution algorithm are the time evolution with the system Hamiltonian, $e^{-iHt}$, and the joint time evolution of the system and bath with the jump operator $e^{-i \sigma \otimes A t}$. To stay within the low-ancilla paradigm we simulate the time evolution via Trotterization, which has asymptotically worse performance than QSVT-based methods \cite{martyn2021grand} but tends to perform much better in practice than predicted by worst-case bounds.
In order to obtain concrete parameters in the rigorous regime, we assign our total error budget of $\varepsilon_\text{Tr} = 0.005 \cdot n_\text{sites}/\norm{H}$ over the different error sources as follows:
\begin{enumerate}
    \item $\varepsilon_{\tau}=70\%$: The budget allocated to the Lindblad time step $\tau$. It controls the error from Lindblad splitting, dilation and quadrature Trotterization, see \cref{prop:total_error}. The overall cost scales linearly with $1/\tau$, justifying a large allocation to this error source.
    \item $\varepsilon_{M_S}=15\%$: The error budget allocated to the finite truncation and quadrature of the filter integral with $M_S$ steps. This is a major source of cost, but larger budgets only reduce $M_S$ polylogarithmically, so we allocate only a moderate amount of budget to it.
    \item $\varepsilon_{H,\text{trott}}=10\%$: The error due to Trotter splitting of the Hamiltonian, see \cref{sec:H_trott_err}. This is controlled by further splitting the evolutions $e^{-i\tau H}$ and $e^{-i\tau_S H}$ into $M_{T,1}$ and $M_{T,2}$ steps, respectively. Since $\tau$ and $\tau_S$ are already small to control other errors, the additional cost is moderate.
    \item $\varepsilon_{A,\text{trott}}=1\%$: The error due to Trotter splitting of $e^{-\sqrt{\tau} \tau_S (X+Y)\otimes A }$, see \cref{sec:A_trott_err}. $A$ does not scale with the system size, while both $\tau$ and $\tau_S$ do, so additional Trotter steps are never needed.
    \item $\varepsilon_{H,\text{rot}}=4\%$: The error due to finite precision of the Clifford+T approximation to rotation synthesis. This error is exponentially suppressed in the number of T gates used for the synthesis and therefore easy to control.
\end{enumerate}

\subsection{Trotter compilation}
\label{sec:trotter_comp}

We first map the fermionic Hamiltonian \cref{eqn:hubabrd_def} to Pauli operators via the Jordan-Wigner mapping \cite{jordan1928paulische}. We use a second-order Trotter product formula, incurring an error that we denote $\varepsilon_{H,\text{trott}}$. We consider a chemically-shifted version of the Hamiltonian \cite{campbell2022early, bay2025fault}, which ensures a simple interaction term of $(U/4) \sum_i Z_{i,\uparrow} Z_{i,\downarrow}$ after performing a Jordan-Wigner transformation, while not changing energy differences of the spectrum within a particle number sector. Because all terms in this sum commute, we can implement the interaction term directly by applying layers of two-qubit $ZZ$-rotations to the system, each of which can be implemented as a single-qubit $Z$ rotation up to Cliffords. Slightly more work is required for the hopping Hamiltonian. Under the Jordan-Wigner mapping the hopping term between sites $i$ and $j$ with spin $\sigma$ becomes
\begin{align}
    c^\dagger_{i,\sigma} c_{j,\sigma} +  c^\dagger_{j,\sigma} c_{i,\sigma} \rightarrow &  \frac{1}{2} (X_{i,\sigma} \otimes \overrightarrow{Z} \otimes X_{j,\sigma} + Y_{i,\sigma} \otimes \overrightarrow{Z} \otimes Y_{j,\sigma} ) ,
\end{align}
where $X_i \otimes \overrightarrow{Z} \otimes X_j = X_i \otimes Z_{i+1}\otimes...\otimes Z_{j-1}\otimes X_j$ and we suppress the spin index. The length of the $\overrightarrow{Z}$ string depends on the fermionic ordering, which can be chosen arbitrarily, with consecutive indices needing no $Z$ terms. For a one-dimensional lattice with open boundary conditions, we take the second-order product formula (for $r$ steps):
\begin{equation}
    e^{-iHtr} \simeq e^{-i H_I t /2} \Big( e^{-i H_h^1 t /2} e^{-i H_h^2 t} e^{-i H_h^1 t /2} e^{-i H_I t} \Big)^r e^{i H_I t/2},
    \label{eqn:H_trotter_step_1d}
\end{equation}
where $H_h^1$ and $H_h^2$ are two sections of the hopping Hamiltonian, chosen in the obvious way such that no fermionic swap gates are needed to avoid implementing $\overrightarrow{Z}$ strings in either section.

While we only perform simulations and report resource estimates for one-dimensional lattices, we will argue below that the simulation cost per unit time is not significantly different for two-dimensional lattices, and our implementation in Qualtran \cite{harrigan2024expressing} focuses on such two-dimensional lattices, and so we briefly describe this. Interaction terms can be implemented exactly as described above. For the hopping term, we follow \cite[Fig. 1(a)]{bay2025fault} and partition the $L_1 \times L_2$ square lattice (with open boundaries) into $2 L_1 L_2 - L_1 - L_2$ $S_1$ tiles, each covering one edge. All hopping terms can be implemented in four Trotter partitions, incurring only a small additional Trotter error compared to the 1D case, and can be implemented locally using only two fermionic orderings, which one can switch between using layers of fermionic swap gates \cite{Babbush2018}. The time evolution for $r$ steps is 
\begin{align}
    e^{-iHtr} \simeq e^{-iH_It/2} \left( e^{-iH_{-,1}t/2} e^{-iH_{-,2}t/2} S e^{-iH_{|,1}t/2} e^{-iH_{|,2}t} e^{-iH_{|,1}t/2} S^\dagger e^{-iH_{-,2}t/2} e^{-iH_{-,1}t/2} e^{-i H_It} 
\right)^r e^{i H_It/2} , \label{eqn:H_trotter_step}
\end{align}
where $S$ denotes the fermionic swap layers. For a square lattice there are slightly more efficient tilings than the $S_1$ approach \cite{campbell2022early, bay2025fault}, but the $S_1$ tiling is simple and easily applied to any lattice.

For both one-dimensional and two-dimensional lattices, the entangling layer is taken as a tensor product of a Pauli operator $\sigma_{X/Y}$ with the jump operator, which we chose to be a randomly-selected hopping term with a purely imaginary coefficient:
\begin{align}
    (a X_\text{anc} + b Y_\text{anc}) \otimes i \im( m)(c^\dagger_{i,\sigma} c_{j,\sigma}-c^\dagger_{j,\sigma} c_{i,\sigma}) \to  (a X_\text{anc} + b Y_\text{anc}) \otimes \frac{\im(m)}{2} (Y_{i,\sigma} \otimes \overrightarrow{Z} \otimes X_{j,\sigma}-X_{i,\sigma} \otimes \overrightarrow{Z} \otimes Y_{j,\sigma}).
\end{align}
By ensuring that the ordering at the end of the last Hamiltonian evolution step has $i$ and $j$ as consecutive labels the $\overrightarrow{Z}$ terms vanish again, and the evolution simplifies to implementing two weight-3 Pauli rotations. We use first-order Trotter splitting to separate the $X_\text{anc}$ and $Y_\text{anc}$ terms, incurring a (comparatively small) Trotter error, $\varepsilon_{A,\text{trott}}$. Together with the Trotter compilation of Hamiltonian evolution described above, this allows us to cost the entire algorithm in the Qualtran framework \cite{harrigan2024expressing}. 

In the standard Pauli-based computation model \cite{bravyi2016trading,litinski2019game}, the key resource is the number of non-Clifford gates required to perform the simulation. These appear in the Trotterized algorithm primarily through Pauli rotations, which can be decomposed into Clifford+T sequences using rotation synthesis. Using the mixed-fallback synthesis method \cite{kliuchnikov2023shorter}, a diamond distance error of $\varepsilon$ per rotation can be achieved with
\begin{align}
    N_{T} \simeq 0.53 \log_2(1/\varepsilon) +4.86
\end{align}
T gates on average per single-qubit rotation, while only requiring a single additional ancilla for the fallback protocol. Two such rotations are needed to implement a single $S_1$ tile, while one such rotation is needed to implement a $ZZ$-Pauli rotation for each interaction term. Distributing the error budget above evenly over all rotations in the circuit, we find that we can always control the error with fewer than 37 T gates per rotation (the rigorous costing with 75 sites). For the empirical scaling we never need more than 26 T gates per rotation and even fewer for the smaller systems. We also considered taking into account the reduction in cost from sampling the identity in the mixed-fallback routine \cite{bothe2026more}, but this yielded only a small improvement of a few percent of the total cost for the examples studied here.

\subsection{Results}

We are now in a position to determine the number of T gates required to converge the simulated systems in \cref{fig:parameter_scaling_simulations}, for which we can read off the time required to converge below the energy error threshold. The results are shown in \cref{tab:small_system_costing}, demonstrating that even small systems can already require billions of T gates to converge to good precision.

\begin{table}[h]
    \renewcommand{\arraystretch}{1.2}
    \setlength{\tabcolsep}{10pt} 
    \centering
    \begin{tabular}{llrrrrr}
\toprule
Scaling & System & $T_\mathrm{mix}$ & $M_S$ & $\tau$ & Steps & T count \\
\midrule
\multirow{3}{*}{Empirical (general)} & $2\times1$ & 2 & 18 & 0.1 & 20 & $7.82 \times 10^{5}$ \\
 & $4\times1$ & 9 & 59 & 0.0345 & 261 & $9.62 \times 10^{7}$ \\
 & $6\times1$ & 19 & 122 & 0.0135 & 1406 & $1.92 \times 10^{9}$ \\
\midrule
\multirow{4}{*}{Empirical (known gap)} & $2\times1$ & 2 & 4 & 0.125 & 16 & $1.46 \times 10^{5}$ \\
 & $4\times1$ & 13 & 13 & 0.0435 & 299 & $2.42 \times 10^{7}$ \\
 & $6\times1$ & 20 & 28 & 0.0161 & 1240 & $3.66 \times 10^{8}$ \\
 & $8\times1$ & 36 & 48 & 0.00685 & 5256 & $3.92 \times 10^{9}$ \\
\bottomrule
\end{tabular}
    \caption{Parameters, mixing times and required number of T gates for the simulated systems in \cref{fig:parameter_scaling_simulations}.}
    \label{tab:small_system_costing}
\end{table}

Beyond the simulated systems we cannot readily predict the mixing time required for convergence, and this will strongly depend on the choice of mixers, system parameters and geometry. We instead calculate the cost to perform a single unit of time evolution. The actual number of steps required for convergence scales quadratically with the mixing time, $M \propto T_\text{mix}^2$, because not only do we have to evolve for longer times, but the error budget is also spread over more steps, necessitating smaller time steps, $\tau$. 

In the case of the empirical scaling we assume that the required $\tau$ does not depend on the system size (see the discussion in \cref{sec:fixed_point_error}), so for a single time unit $\tau$ can be taken as constant. From the numerical results we estimate somewhat conservatively that $\tau = 0.15$ for the general empirical scaling and $\tau=0.25$ for the empirical scaling with a known gap, see \cref{sec:spectral_scaling}.

The resulting T-gate counts for different system sizes are shown in \cref{fig:T_scaling}. Compared to the rigorous scalings we see that the empirical parameter scalings require orders of magnitude lower T-gate counts, although they remain significant with $10^8-10^9$ T gates in the range of system sizes presented. However, the scaling with the size of the system reduces to approximately linear in the number of sites for larger systems, where the extensive error budget loosens the requirements on the filter function. 

\begin{figure}
    \centering
    \includegraphics[scale=1.0]{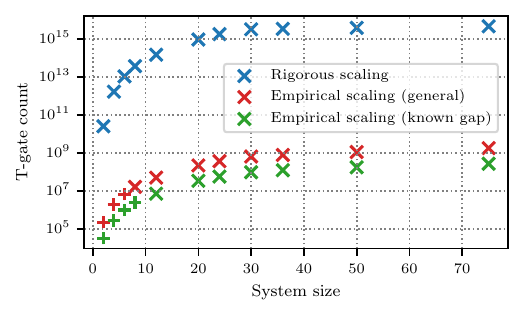}
    \caption{Estimated number of T gates required to simulate linear Hubbard models with open boundary conditions of various sizes for one unit of time. The $+$ markers refer to simulated systems and the $\times$ markers to extrapolated sizes.}
    \label{fig:T_scaling}
\end{figure}

While we have only simulated and given resource estimates for 1D lattices so far, it seems plausible that the costs per unit of time for an extensive energy budget will be very similar for 2D systems with the same number of sites, because the geometry of the lattice only enters our costing through the number of edges, $n_\text{edge}$, which are at most twice as many in a 2D lattice with the same number of sites (depending on whether periodic boundaries conditions are applied or not). Firstly, the upper bound on the Hamiltonian spectral norm, \cref{eqn:H_spectral_bound}, will at most increase by a constant factor proportional to the number of edges. Secondly, the cost of the $H$-evolution again grows at most linearly with the number of edges, as the $S_1$ tiling considered here can be directly applied to more complex lattices. Finally, the worst-case Trotter error constant for evolution with $H$, \cref{eqn:general_trotter_constant}, also grows linearly with the number of edges \cite{campbell2022early}. The T gate cost of simulating a 2D Hubbard model for a single time unit should therefore only differ from the results shown here by a modest constant prefactor.

\begin{table}[h]
\renewcommand{\arraystretch}{1.2}
\setlength{\tabcolsep}{10pt}
\setlength{\aboverulesep}{0pt}
\setlength{\belowrulesep}{0pt}
\centering
\begin{tabular}{lcc}
\toprule
 & \multicolumn{2}{c}{T-gate count} \\
\cmidrule(l){2-3}
Regime & 6 sites & 36 sites \\
\midrule
Rigorous              & $1.0 \times 10^{13}$ & $3.3 \times 10^{15}$ \\
Empirical (general)   & $6.6 \times 10^{6}$  & $7.7 \times 10^{8}$ \\
Empirical (known gap) & $9.8 \times 10^{5}$  & $1.2 \times 10^{8}$ \\
\bottomrule
\end{tabular}
\caption{Estimates of fault-tolerant resources required to simulate a single time unit of Lindblad evolution for 6-site and 36-site linear Hubbard models. The number of logical qubits required, 14 and 74 respectively, is the same for all parameter regimes.}
\label{tab:T_counts}
\end{table}

\Cref{tab:T_counts} shows  that for a 36-site system we require more than $7.7\times10^{8}$ T gates, but the total cost to convergence can be much higher once the mixing time is taken into account. The numerical evidence here and in other papers \cite{zhan2026rapid} suggests that a single time unit is likely an underestimate of the total mixing time. In fact, the much smaller 1D 8-site model already requires around 30 time units for convergence. It is known that the scaling of the mixing time also depends on the choice of jump operators \cite{zhan2026rapid}, with system-spanning operators being capable of inducing rapid mixing in some cases, although at additional cost per step. We did not fully explore different mixers and so faster convergence times are possible. However, even in the case of rapid mixing, the mixing time still scales logarithmically with the system size, and so a single time unit is likely a conservative benchmark.

The cost is compounded by the fact that the required Hamiltonian simulation time scales faster than linearly with the total Lindblad simulation time: Not only are more simulation steps needed if the Lindblad simulation time is longer, each simulation step has to be performed with greater accuracy in order to stay within the total error budget.  This raises the question of what drives up the resource costs: \Cref{fig:evolution_time_breakdown} shows that the overall algorithm is dominated by the Hamiltonian simulation required to approximate the filter integral, \cref{eqn:W_definition}. Unfortunately, these evolution times are crucial to achieve sufficiently low fixed-point energy error. At an abstract level, the algorithm has to resolve the energy differences on the scale $S_\omega / \Delta$ for each dissipative step to ensure a decrease in energy. 

\begin{figure}
    \centering
    \includegraphics[scale=1.0]{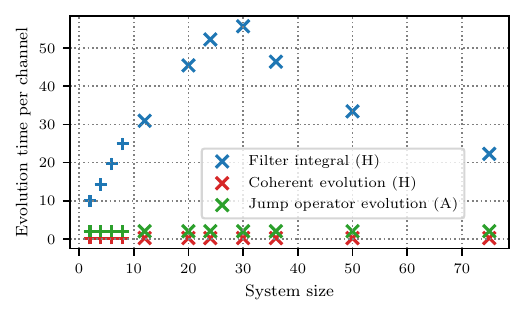}
    \caption{Share of the different primitives in the total unitary evolution for the empirical scaling with known gap: The $+$ markers refer to simulated systems and the $\times$ markers to extrapolated sizes. The $H$ evolution for the filter dominates the overall cost, with evolution by $A$ being much smaller. The required evolution time initially increases as the spectral gap becomes smaller, and then decays for larger systems as the targeted extensive energy error, $\varepsilon = 0.005 \cdot n_\text{sites}$, takes over, which puts less stringent demands on the filter-integral approximation.}
    \label{fig:evolution_time_breakdown}
\end{figure}

\section{Conclusion}

Our analysis bridges the gap between the theoretical ideas behind Lindblad-based ground-state preparation and the practical costs of running these algorithms on error-corrected quantum hardware. In addition to providing rigorous upper bounds, we combine scaling arguments with prefactors obtained from concrete gate‑level resource estimates to give a clearer picture of what will be required to make Lindblad state preparation work in practice.

We identify the spectral filtering integral as the most expensive subroutine in the algorithm, dominating the overall runtime and resource costs. While the asymptotic scaling of the filter cost is bounded from below by the time-bandwidth uncertainty, we found that there exists a speed-precision trade-off, which allows one to trade higher energy error (more leakage) for faster convergence times (stronger dissipation).

We further study the fermionic Hubbard model as a benchmark system and compile the Lindblad evolution down to the gate level for concrete resource estimation.
Through numerical simulations we identify optimal parameters that improve on the theoretical error bounds by orders of magnitudes and allow us to predict the resources required for classically non-trivial systems. While our approach has only moderate qubit requirements we find that the T-gate cost of running even a single time unit of Lindblad evolution is substantial, often much higher than the subsequent quantum phase estimation routine, for which this procedure might provide the initial state. 

The bottleneck of implementing the Hamiltonian evolution for the filter integral points to a number of potential avenues for improvement that have been developed in the literature for other costing models. One such avenue is using locality for spatial truncation \cite{hahn2026efficient,mizuta2026lattice}, although these works focus on systems where locality holds at the qubit level. Further optimizations of the energy filter are discussed in \cite{elman2026localityfilterdesigndissipative}, together with a proposal for exploiting locality in a block-encoded evolution of the Hubbard model. An interesting future direction would be to apply these techniques to the Trotter-based approach of this paper.
The comparatively low qubit requirements also make it feasible to explore other simulation techniques that require more than one ancilla, as outlined in \cite{ding2024simulating,cleve1612efficient,kato2026}. Together, these results suggest a path for practical implementations of the Lindblad-based ground-state preparation on fault-tolerant quantum computers.

\section{Acknowledgments}

M.B. conceived the project, derived the error bounds and ran the simulations. P.S. provided supervision and helped with the writing of the manuscript. N.B. provided the compilation of the fermionic model, helped with the interpretation of results and contributed extensively to the writing of the manuscript.
M.B. is a Sustaining Innovation Postdoctoral Research Associate at Astex Pharmaceuticals and thanks Astex Pharmaceuticals for funding. AI was used for debugging code and for reviewing the manuscript and the derivations.

\bibliography{main}

\appendix

\crefalias{section}{appendix}

\section{Properties of the trace norm}
\label{sec:trace_norm}

In this appendix we summarize some useful properties of the trace norm, defined
\begin{align}
\tnorm{A} = \Tr\sqrt{A^\dagger A} = \sum_i \sigma_i(A), \quad \text{}
\end{align}
where $A$ is a linear operator and $\{ \sigma_i(A) \}$ are the singular values of A.

\begin{itemize}
    \item Unitary invariance: For unitary $U$
\begin{align}
    \tnorm{U^\dagger AU} = \tnorm{A}
\end{align}
\item Contractivity \cite{ruskai1994beyond}: For any CPTP map $\Phi$ and any density operators $\rho$, $\sigma$
\begin{align}
    \tnorm{\Phi(\rho-\sigma)} \leq \tnorm{\rho-\sigma}
\end{align}
\item Spectral bound: 
\begin{align}
    \tnorm{A} \geq \norm{A},
\end{align}
where $\norm{\, \cdot \,}$ denotes the spectral norm.
\item Hölder's inequality (submultiplicativity):
\begin{align}
    \tnorm{AB}\leq \tnorm{A} \norm{B}.\label{eqn:submultiplicativity}
\end{align}
Combined with $\tnorm{A^\dagger} = \tnorm{A}$ this can be applied on either side of the argument, so that $\tnorm{A B C} \leq \norm{A} \tnorm{B} \norm{C}$.
\end{itemize}

For the error bounds we also need to measure the size of superoperators, for which we use the norm that the trace norm induces on them,
\begin{align}
    \norm{\mathcal{E}}_1 := \sup_{X = X^\dagger \neq 0} \frac{\tnorm{\mathcal{E}[X]}}{\tnorm{X}} . \label{eqn:induced_norm}
\end{align}
We denote this by the same symbol as the trace norm; it is to be read as \cref{eqn:induced_norm} whenever the argument is a superoperator. Restricting the supremum to Hermitian $X$ is sufficient for our purposes, because all superoperators we consider map Hermitian operators to Hermitian operators. This norm inherits two properties that we use repeatedly:
\begin{itemize}
    \item Submultiplicativity under composition, which holds for any induced norm,
    \begin{align}
        \norm{\mathcal{E}_1 \mathcal{E}_2}_1 \leq \norm{\mathcal{E}_1}_1 \norm{\mathcal{E}_2}_1 . \label{eqn:induced_submult}
    \end{align}
    This is what allows us to separate products of superoperators, noting that the product denotes composition rather than an inner product.
    \item Contractivity of Lindbladian evolution: any Lindbladian $\Lind$ generates a CPTP semigroup, so decomposing a Hermitian $X = X_+ - X_-$ into positive and negative parts of orthogonal support and using positivity and trace preservation gives $\tnorm{e^{\Lind t}[X]} \leq \tnorm{e^{\Lind t}[X_+]} + \tnorm{e^{\Lind t}[X_-]} = \Tr X_+ + \Tr X_- = \tnorm{X}$, and therefore
    \begin{align}
        \norm{e^{\Lind t}}_1 \leq 1 . \label{eqn:semigroup_contractive}
    \end{align}
\end{itemize}

\section{Useful integral formulas}
\label{sec:duhamel}

The approach described in this section is the main building block for deriving commutator bounds \cite{childs2021theory}. Consider a contractive evolution induced by a (super-) operator, $e^{\Lind (t-t_0)}$ with $t>t_0$, where the exponential is defined in the power series with compositions of superoperators as the multiplication. We can consider the first derivative of the exponential
\begin{align}
    \partial_t e^{\Lind (t-t_0)} = e^{\Lind (t-t_0)}\Lind \,
\end{align}
and by integration we obtain an integral formulation of the original evolution
\begin{align}
    e^{\Lind (t-t_0)} = \int_{t_0}^t  e^{\Lind (t_1-t_0)} \Lind dt_1  + \mathbb{1}.
\end{align}
Repeating this process with $n$'th-order derivatives we obtain
\begin{align}
    e^{\Lind (t-t_0)} = \int_{t_0}^t dt_n \,\int_{t_0}^{t_{n}} dt_{n-1}... \int_{t_0}^{t_{2}} dt_1 \, e^{\Lind (t_1-t_0)} (\Lind)^n + \sum_{l=0}^{n-1} \frac{(t-t_0)^l}{l!}(\Lind)^l.
\end{align}
The second term here is just the Taylor series of the exponential to order $n-1$, so it follows that all higher-order terms in the series are captured in the $n$-fold integration. This is very useful, because inside the trace norm we can easily evaluate the integral:
\begin{align}
    \tnorm{e^{\Lind (t-t_0)}[\rho]-\sum_{l=0}^{n-1} \frac{(t-t_0)^l}{l!}(\Lind)^l[\rho]} &= \tnorm{\int_{t_0}^t dt_n \,\int_{t_0}^{t_{n}} dt_{n-1}... \int_{t_0}^{t_{2}} dt_1 \, e^{\Lind (t_1-t_0)} (\Lind)^n [\rho]}\\
    &\leq \int_{t_0}^t dt_n \,\int_{t_0}^{t_{n}} dt_{n-1}... \int_{t_0}^{t_{2}} dt_1 \, \tnorm{e^{\Lind (t_1-t_0)} (\Lind)^n[\rho]}\\
    &\leq \int_{t_0}^t dt_n \,\int_{t_0}^{t_{n}} dt_{n-1}... \int_{t_0}^{t_{2}} dt_1 \, \tnorm{(\Lind)^n[\rho]}\\
    &= \frac{(t-t_0)^n}{n!} \tnorm{(\Lind)^n[\rho]} .
\end{align}

In addition we frequently make use of Duhamel's principle (variation of constants), which states that
\begin{align}
    \dot{x}(t) = \mathcal{L} x(t) + g(t) \quad \Rightarrow \quad 
    x(t) = e^{\mathcal{L} (t-t_0)} x(t_0) + \int_{t_0}^t e^{\mathcal{L} (t-s)} g(s) ds .
\end{align}
In the following the case where $x(t_0)=0$ is most relevant, because it allows us to extract the leading-order time dependence by rewriting a term that vanishes as $t\to t_0$ as a time integral. 

\section{Summary of the algorithm}
\label{sec:algorithm_description}
The algorithm in \cite{ding2024single} is defined as a series of steps
\begin{align}
    \rho_m \to \rho_{m+1} &= e^{-i H \tau} \Tr_a \left[ W(\sqrt{\tau}) ( \ket{0} \bra{0} \otimes \rho_m ) W^\dagger(\sqrt{\tau}) \right]  e^{i H \tau} \\
    &= e^{-i H \tau} \mathcal{W}(\tau)[\rho_{m}] e^{i H \tau} ,
\end{align}
where $W(\sqrt{\tau})$ is the Trotterized unitary
\begin{align}
     W(\sqrt{\tau})  := e^{iH S_S} \prod_{l = -M_S}^{M_S} e^{-\frac{i\sqrt{\tau}}{2} \tilde{H}_l} \prod_{l=M_S}^{-M_S} e^{-\frac{i\sqrt{\tau}}{2} \tilde{H}_l} e^{-iHS_S} , \label{eqn:W_definition}
\end{align}
where $S_S = M_S \tau_S$ is the time evolution needed for the truncated filter integral, which runs over $[-S_S, S_S]$. $\mathcal{W}(\tau)$ is the corresponding channel
\begin{align}
    \mathcal{W}(\tau)[\rho_{m}] := \Tr_a \left[ W(\sqrt{\tau}) ( \ket{0} \bra{0} \otimes \rho_m ) W^\dagger(\sqrt{\tau}) \right] ,
\end{align}
and $\tilde{H}_l$ is the $l$'th term in the dilated approximate transition operator:
\begin{align}
 \tilde{H}_l &:=
    \begin{pmatrix}
            0 & ( f(s_l) w_l A(s_l))^\dagger \\
            f(s_l) w_l A(s_l) & 0 
        \end{pmatrix} \\
        &= w_l (\sigma_x  \re f(s_l) + \sigma_y \im f(s_l)) A(s_l)\\
        &=: \sigma_l \otimes A(s_l) ,
\end{align}
where we assume that $A$ is Hermitian. In practice we can implement $W(\sqrt{\tau})$ with shorter Hamiltonian evolution by noting that it is equal to 
\begin{align}
W(\sqrt{\tau})=& \left(
\overrightarrow{\prod_{\ell}}
\, \tilde{A}_{\ell}(\sqrt{\tau})
\bigl( I \otimes e^{i H \tau_S} \bigr)
\right)
\left(
\overleftarrow{\prod_{\ell}}
\bigl( I \otimes e^{-i H \tau_S} \bigr)
\tilde{A}_{\ell}(\sqrt{\tau})
\right) \\
\tilde{A}_\ell(\sqrt{\tau}) :=& e^{-i \frac{\sqrt{\tau}}{2} \sigma_\ell \otimes A} .
\end{align}
The algorithm can therefore be implemented with (approximate) access to $e^{-iHt}$, $\tilde{A}$ and a measure-reset operation. 

\section{Comparison to other algorithms}

There exists a variety of simulation algorithms in the literature, which we compare to the single-ancilla algorithm \cite{ding2024single}. 
The block-encoding-based higher-order schemes from \cite{cleve1612efficient} have higher ancilla overhead and therefore may be less suitable for early applications, so we do not consider them here.

In \cite{ding2024single} the authors also present a second, discrete-time algorithm that has more favorable scaling in the mixing time but doesn't track the exact evolution of the underlying Lindblad operator, which is why we do not consider it in the rigorous regime. A follow-up work by overlapping authors \cite{zhan2026rapid} also focuses on continuous-time dynamics. The empirical regime makes the exact same assumptions as the discrete algorithm, but does not implement additional Trotter splitting of $W(\tau)$, to stay consistent with the rigorous regime. It would therefore be interesting to see the effect of the additional Trotterization on our empirical results. For the single-time-unit costing this is unlikely to make a difference, since the time steps there are already very large, but it would likely improve the scaling with $T_\text{mix}$.

Another single-ancilla algorithm for both thermal and ground-state preparation is presented in~\cite{ding2025end}. Rather than simulating a prescribed Lindbladian, it repeatedly evolves the system together with a single bath qubit under a time-dependent system-bath Hamiltonian, then resets the bath. The effective Lindbladian emerges at second order in the coupling strength $\alpha$, so that $\alpha^2$ plays the role of our time step $\tau$. This differs from the approach studied here in two main respects. First, the coherent evolution is part of the joint system-bath evolution rather than split off as a separate $e^{-iH\tau}$ step, so only forward Hamiltonian evolution is needed. The price is a Lamb-shift term in the effective dynamics. Second, the filter is a Gaussian that is narrow in frequency, with width ${\sim}\sigma^{-1}$ and $\sigma \sim \Delta^{-1}$. Each channel application therefore only drives transitions near a randomly sampled bath frequency $\omega$. Our filter, by contrast, covers the full range of energy-decreasing transitions $[-S_\omega, 0]$ in every step. The authors of~\cite{ding2025end} note that their narrow filter simplifies the fixed-point analysis, but may increase the total simulation time compared to broader filters as in~\cite{ding2024single}. Because the filter is generated by a time-ordered evolution rather than a quadrature sum inside a single exponential, the error contributions of the two schemes do not map one-to-one. 

Shortly before submission of this work, Elman submitted independent work~\cite{elman2026localityfilterdesigndissipative} which reduces the cost of the filter integral in two ways: by optimizing the filter coefficients on a fixed grid of evolution times, which can be cast as a linear program, and by exploiting locality to restrict the Hamiltonian evolution in the filter to a patch of sites around the coupling operator, which makes the per-branch cost independent of system size at fixed patch radius. 
That work includes density-matrix simulations of the Lindblad dynamics for a range of molecules as well as Ising and Hubbard chains, together with fault-tolerant cost estimates for a single channel application.
Ref.~\cite{elman2026localityfilterdesigndissipative} likewise observes that the high non-Clifford gate counts of single-ancilla dissipative schemes temper near-term prospects, and its improvements target the filter integral that we identify as the dominant cost.
The two approaches differ in several respects.
The simulations in~\cite{elman2026localityfilterdesigndissipative} evolve the Lindblad equation directly rather than the compiled circuit, and therefore do not capture the dilation and Lindblad-splitting errors at finite time step, which, while leaving the fixed point of an ideal filter unchanged (see Appendix~E7), can alter the mixing time.
Its cost estimates are based on block-encoded Hamiltonian simulation rather than Trotterization, which requires additional ancilla qubits, and are reported per channel application.
Applying the techniques of~\cite{elman2026localityfilterdesigndissipative} in the Trotterized setting considered here is a promising direction for future work.

\section{Error bounds}
\label{sec:errors}

The algorithm approximates the exact Lindbladian simulation up to a unitary transformation\footnote{We do not need to simulate this rotation because it has no impact on the final-state energy.}
\begin{align}
    \bar{\rho}_{M_T} &=  e^{-i H S_S} \rho_{M_T} e^{i H S_S} \simeq \rho(T) . \
\end{align}
The rotated density matrix follows the same dynamics as the original one
\begin{align}
        \bar{\rho}_{m+1} &= e^{-i H \tau} \Tr_a \left[e^{-i H S_S}  W(\sqrt{\tau}) e^{i H S_S} ( \ket{0} \bra{0} \otimes \bar{\rho}_m ) e^{-i H S_S} W^\dagger(\sqrt{\tau}) e^{i H S_S} \right] e^{i H \tau} \\
        &= e^{-i H \tau} \Tr_a \left[ \prod_{l = -M_S}^{M_S} e^{-\frac{i\sqrt{\tau}}{2} \tilde{H}_l} \prod_{l=M_S}^{-M_S} e^{-\frac{i\sqrt{\tau}}{2} \tilde{H}_l}  ( \ket{0} \bra{0} \otimes \bar{\rho}_m ) \prod_{l = -M_S}^{M_S} e^{\frac{i\sqrt{\tau}}{2} \tilde{H}_l} \prod_{l=M_S}^{-M_S} e^{\frac{i\sqrt{\tau}}{2} \tilde{H}_l} \right] e^{i H \tau} \\
        &=: e^{-i H \tau} \mathcal{W}(\tau)[\bar{\rho}_{m}] e^{i H \tau} .
\end{align}
The aim is therefore to bound the difference between the approximate density matrix $\bar{\rho}_M$ and the exact Lindblad evolution for time $T=M\tau$ in the trace norm\footnote{The differences in the energy that we use can be directly derived from this via the submultiplicativity of the trace norm: One just has to multiply the bound by the spectral norm of the Hamiltonian.}, i.e.
\begin{align}
    \tnorm{\rho(M\tau) - \bar{\rho}_M} \leq M \tnorm{\rho(\tau)-\hat{\rho}_1} ,
\end{align}
where the inequality follows from the standard contractivity argument \cite{ruskai1994beyond}, meaning that it is sufficient to analyze the error in a single step and then extrapolate linearly.
The total simulation error arises from a series of approximations:
\begin{enumerate}
    \item Lindblad splitting error: 
    \begin{align}
        e^{\Lind \tau} \rho \simeq e^{\Lind_H \tau} e^{\Lind_K \tau} \rho
    \end{align}
    \item Dilation error, meaning the error from approximating the split Lindblad dynamics by the dilation step:
    \begin{align}
         \tilde{\rho}_{1} = e^{-iH \tau} \Tr_a (e^{-i \sqrt{\tau}\tilde{K}} \ket{0} \bra{0} \otimes \rho(0) e^{i \sqrt{\tau}\tilde{K}})e^{iH \tau} 
    \end{align}
    \item Quadrature error, from approximating the integral in the transition operator $K$ by quadrature 
    \begin{align}
    K &= \int_{-\infty}^{\infty} f(s) A(s) ds, \quad \tilde{K} = \begin{pmatrix}
            0 & K^\dagger \\
            K & 0 
            \end{pmatrix}\\
    K_S &= \sum_{l=-M_S}^{M_S} w_l f(s_l) e^{iHs_l} A e^{-iHs_l}\\
        \tilde{K}_S &= \begin{pmatrix}
            0 & ( \sum_{l=-M_S}^{M_S} f(s_l) w_l A(s_l))^\dagger \\
            \sum_{l=-M_S}^{M_S} f(s_l) w_l A(s_l) & 0 
        \end{pmatrix} := \sum_{l=-M_S}^{M_S} \tilde{H}_l
    \end{align}
    with density matrix
    \begin{align}
        \tilde{\rho}_{1,s} = e^{-iH \tau} \Tr_a (e^{-i \sqrt{\tau}\tilde{K}_s} \ket{0} \bra{0} \otimes \rho(0) e^{i \sqrt{\tau}\tilde{K}_s})e^{iH \tau} \simeq \tilde{\rho}_1
    \end{align}
    \item Quadrature Trotterization error, from splitting the sum of quadrature points into different exponentials:
    \begin{align}
        e^{-i \sqrt{\tau} \tilde{K}_S } \simeq e^{-iH S_S} W(\sqrt{\tau}) e^{iH S_S} := \prod_{l = -M_S}^{M_S} e^{-\frac{i\sqrt{\tau}}{2} \tilde{H}_l} \prod_{l=M_S}^{-M_S} e^{-\frac{i\sqrt{\tau}}{2} \tilde{H}_l}
    \end{align}
    \item Trotter error: The error incurred from implementing the evolution $A$ and $H$ with a first- and second- order Trotter formula, respectively.
    \begin{align}
        \bar{\rho}_{1,\text{trott}} = \bar{\rho}|_{\tilde{A}\to S_1(\tilde{A}), e^{-iHt}\to S_2(e^{-iHt}) }
    \end{align}
\end{enumerate}

Using the triangle inequality these errors add up linearly:
\begin{align}
\begin{split}
    \norm{\rho(\tau) - \bar{\rho}_{1,\text{trott}}}_1 \leq & \underbrace{\norm{e^{(\mathcal{L}_H + \mathcal{L}_K) \tau}\rho(0) - e^{\mathcal{L}_H \tau} e^{\mathcal{L}_K \tau} \rho(0)}_1}_{\text{Lindblad splitting error}} 
    + \underbrace{\norm{e^{\mathcal{L}_H \tau} e^{\mathcal{L}_K \tau} \rho(0) - \tilde{\rho}_1}_1}_{\text{Dilation error}} \\
    & + \underbrace{\norm{\tilde{\rho}_1-\tilde{\rho}_{1,s}}_1}_{\text{Integral quadrature error}} 
    + \underbrace{\norm{\tilde{\rho}_{1,s}-\bar{\rho}_{1}}_1}_{\text{Time sum Trotterization error}} 
    +\underbrace{\tnorm{\bar{\rho}_{1}-\bar{\rho}_{1,\text{trott}}}}_{\text{Trotter evolution error}}.
\end{split}
\end{align}

The result for these errors is restated here: 

\begin{proposition_restatement}
    Let $\bar{\rho}_{M,\text{trott}}$ be the state after $M$ steps of the fully compiled algorithm with step duration $\tau$, total time $T = M\tau$, quadrature spacing $\tau_S$ with $2M_S+1$ points, and Trotter step sizes $\tau_{T,1} = \tau/M_{T,1}$ and $\tau_{T,2} = \tau_S/M_{T,2}$ for the coherent and dissipative $H$ evolutions. Then for any initial state $\rho(0)$
    \begin{align}
    \begin{split}
        \tnorm{\rho(T) - \bar{\rho}_{M,\text{trott}}} \leq T \Big[ & \tau \left( C_S + C_D + C_{T,1} \right) + \tau^2 C_{T,2} + C_Q \\
        & + \Big( 2 \tau_{T,1}^2 + 4 (2M_S+1) \frac{\tau_S}{\tau} \tau_{T,2}^2 \Big) C_H + \tau_S^2 \, C_{\tilde{A}} \Big] ,
    \end{split}
    \end{align}
    with the constants
    \begin{align}
        C_S &:= 4 \norm{H} \norm{A}^2 F_\text{tot}^2 , \\
        C_D &:= \frac{8}{3} \norm{A}^4 F_\text{tot}^4 , \\
        C_Q &:= 4 \norm{A}^2 F_\text{tot} F(M_S) , \\
        \begin{split}
    C_{T,1} &:=  \sum_{k=1}^{K}  \left[ \frac{1}{24} \norm{\ad^3_{V_{k}}(v_{k})} + \frac{1}{96} \norm{\ad^3_{v_{k}}( V_{k})} \right. \\
        & \quad \left. + \frac{1}{24}  \Big\| \ad^2_{V_k} v_k
        -  \frac{1}{2}\ad^2_{v_k} V_k \Big\|
        \Big( 11 \norm{V_{k+1}} + 3 \norm{V_k} + 3 \norm{v_k} + 8\sum_{k'=k+1}^K \norm{v_{k'}} \Big)\right] ,
    \end{split}\\
        C_{T,2} &:= \sum_{k=1}^{K}  \frac{1}{144} \left( \norm{\ad^2_{V_{k}} (v_{k})} +  \frac{1}{2} \norm{\ad^2_{v_{k}}(V_{k})} \right)^2  ,\\
        C_H &:= \frac{1}{12} \sum_{\gamma_1=1}^{\Gamma} \left\| \left[ \sum_{\gamma_3=\gamma_1+1}^{\Gamma} H_{\gamma_3}, \left[ \sum_{\gamma_2=\gamma_1+1}^{\Gamma} H_{\gamma_2}, H_{\gamma_1} \right] \right] \right\| + \frac{1}{24} \sum_{\gamma_1=1}^{\Gamma} \left\| \left[ H_{\gamma_1}, \left[ H_{\gamma_1}, \sum_{\gamma_2=\gamma_1+1}^{\Gamma} H_{\gamma_2} \right] \right] \right\| , \\
        C_{\tilde{A}} &:= \norm{A}^2 \sum_{l=-M_S}^{M_S} \left| \operatorname{Im} f(s_l) \right| \left| \operatorname{Re} f(s_l) \right| .
    \end{align}
    Here $F(M_S)\propto \exp(-\sqrt{aM_S})$ is the filter tail mass bounded in \cref{eqn:filter_tail}, $F_\text{tot} = \sum_{l=-\infty}^{\infty} |f(s_l)|\tau_S$, $v_k$, $V_k$ and $K = 2M_S+1$ are defined in \cref{prop:quad_trott_error}, and $C_H$ is the second-order Trotter constant of the system Hamiltonian, $H=:\sum_{\gamma=1}^{\Gamma} H_\gamma$.
\end{proposition_restatement}

\begin{proof}
    By contractivity of the single-step channels the $M$-step error is at most $M$ times the single-step error \cref{eqn:error_overview}. 
    The single-step contributions are $\tau^2 C_S$ (\cref{prop:splitting_error}), $\tau^2 C_D$ (\cref{prop:dilation_error}), $\tau C_Q$ (\cref{prop:filter_tail}), $\tau^2 C_{T,1} + \tau^3 C_{T,2}$ (\cref{prop:quad_trott_error} at $t = \sqrt{\tau}$), $2\tau \tau_{T,1}^2 C_H + 4(2M_S+1) \tau_S \tau_{T,2}^2 C_H$ (\cref{prop:H_trotter_error}) and $\tau \tau_S^2 C_{\tilde{A}}$ (\cref{prop:A_trott_err}). Multiplying by $M = T/\tau$ gives \cref{eqn:total_error}. 
\end{proof}

As discussed in \cref{sec:dilation_error}, this result extends to the statistical sampling of multiple jump operators in the Lindbladian, $\mathcal{L}_K = \sum_k p_k \mathcal{L}_{K,k}$, which can be simulated by classically sampling the jump operators.

\subsection{Lindblad splitting error}
\label{sec:splitting_error}

The coherent and the dissipative part of the Lindbladian in \cref{eqn:lindblad_def} do not commute, so splitting the single-step evolution into the two exponentials $e^{\Lind_H \tau} e^{\Lind_K \tau}$ incurs a Trotter-like error, which we bound in the induced superoperator norm \cref{eqn:induced_norm}.

\begin{proposition}
\label{prop:splitting_error}
    Let $\Lind_H$ and $\Lind_K$ be the coherent and dissipative parts of \cref{eqn:lindblad_def} with a single transition operator $K$. Then for every $t \geq 0$
    \begin{align}
        \norm{e^{(\Lind_H + \Lind_K) t} - e^{\Lind_H t} e^{\Lind_K t}}_1
        \leq \frac{t^2}{2} \norm{\left[ \Lind_H, \Lind_K \right]}_1
        \leq 4 t^2 \norm{H} \norm{K}^2
        \leq t^2 C_S ,
    \end{align}
    with the splitting constant
    \begin{align}
        C_S := 4 \norm{H} \norm{A}^2 F_\text{tot}^2 . \label{eqn:C_S_def}
    \end{align}
\end{proposition}

\begin{proof}
    Abbreviate $\Lind_1 := \Lind_H$ and $\Lind_2 := \Lind_K$. Both are Lindbladians, as is their sum, so \cref{eqn:semigroup_contractive} applies to all three semigroups appearing below.

    We first reduce the difference to a single commutator: The split evolution $g(t) := e^{\Lind_1 t} e^{\Lind_2 t}$ satisfies
    \begin{align}
        g'(t) = \Lind_1 g(t) + g(t) \Lind_2 = \left( \Lind_1 + \Lind_2 \right) g(t) + \left[ g(t), \Lind_2 \right] ,
    \end{align}
    i.e.\ the same linear equation as the exact evolution up to the inhomogeneity $[g(t), \Lind_2]$, which is a known function of $t$ because $g$ is known explicitly. With $g(0) = \mathbb{1}$, Duhamel's principle (\cref{sec:duhamel}) gives
    \begin{align}
        g(t) = e^{(\Lind_1 + \Lind_2) t} + \int_0^t d\tau_1 \, e^{(\Lind_1 + \Lind_2)(t - \tau_1)} \left[ g(\tau_1), \Lind_2 \right] ,
    \end{align}
    so the splitting error is given by the integral alone. Since $e^{\Lind_2 \tau_1}$ commutes with its own generator the inhomogeneity factorizes, $[ g(\tau_1), \Lind_2 ] = [ e^{\Lind_1 \tau_1}, \Lind_2 ] \, e^{\Lind_2 \tau_1}$, and pulling the integral out of the norm and stripping the two semigroup factors with \cref{eqn:induced_submult,eqn:semigroup_contractive} leaves
    \begin{align}
        \norm{e^{(\Lind_1 + \Lind_2)t} - e^{\Lind_1 t} e^{\Lind_2 t}}_1 \leq \int_0^t d\tau_1 \, \norm{\left[ e^{\Lind_1 \tau_1}, \Lind_2 \right]}_1 . \label{eqn:splitting_single_int}
    \end{align}

    Setting $h(t) := [ e^{\Lind_1 t}, \Lind_2 ]$ and using that $\Lind_1$ commutes with $e^{\Lind_1 t}$,
    \begin{align}
        h'(t) = \Lind_1 e^{\Lind_1 t} \Lind_2 - \Lind_2 \Lind_1 e^{\Lind_1 t} = \Lind_1 h(t) + \left[ \Lind_1, \Lind_2 \right] e^{\Lind_1 t} ,
    \end{align}
    now with $h(0) = [ \mathbb{1}, \Lind_2 ] = 0$, so the homogeneous term drops out entirely and
    \begin{align}
        h(\tau_1) = \int_0^{\tau_1} d\tau_2 \, e^{\Lind_1 (\tau_1 - \tau_2)} \left[ \Lind_1, \Lind_2 \right] e^{\Lind_1 \tau_2} .
    \end{align}
    Stripping the two exponentials as before gives $\norm{h(\tau_1)}_1 \leq \tau_1 \norm{[ \Lind_1, \Lind_2 ]}_1$, and inserting this into \cref{eqn:splitting_single_int} yields the first inequality of the proposition,
    \begin{align}
        \norm{e^{(\Lind_1 + \Lind_2)t} - e^{\Lind_1 t} e^{\Lind_2 t}}_1 \leq \int_0^t \!\! d\tau_1 \int_0^{\tau_1} \!\! d\tau_2 \, \norm{\left[ \Lind_1, \Lind_2 \right]}_1 = \frac{t^2}{2} \norm{\left[ \Lind_1, \Lind_2 \right]}_1 .
    \end{align}

    We now bound the commutator: By the triangle inequality and \cref{eqn:induced_submult} it suffices to bound the generators separately, $\norm{[ \Lind_1, \Lind_2 ]}_1 \leq 2 \norm{\Lind_1}_1 \norm{\Lind_2}_1$. Using the two-sided submultiplicativity of the trace norm (\cref{sec:trace_norm}),
    \begin{align}
        \tnorm{\Lind_H[\rho]} &= \tnorm{-i \left[ H, \rho \right]} \leq 2 \norm{H} \tnorm{\rho} , \\
        \tnorm{\Lind_K[\rho]} &= \tnorm{K \rho K^\dagger - \tfrac{1}{2} K^\dagger K \rho - \tfrac{1}{2} \rho K^\dagger K} \leq \left( 1 + \tfrac{1}{2} + \tfrac{1}{2} \right) \norm{K}^2 \tnorm{\rho} ,
    \end{align}
    so $\norm{\Lind_H}_1 \leq 2\norm{H}$ and $\norm{\Lind_K}_1 \leq 2\norm{K}^2$, giving
    \begin{align}
        \norm{\left[ \Lind_H, \Lind_K \right]}_1 \leq 8 \norm{H} \norm{K}^2 \label{eqn:lind_commutator_bound}
    \end{align}
    and hence the second inequality. The third follows from $\norm{K} \leq \norm{A} F_\text{tot}$, which follows from the Poisson identity \cref{eqn:poisson} and the unitary invariance of the spectral norm.
\end{proof}

\subsection{Dilation error}
\label{sec:dilation_error}

Replacing the split Lindblad dynamics by a single dilation step is accurate to second order in the step duration, with the following prefactor.

\begin{proposition}
\label{prop:dilation_error}
    Let $\tilde{\rho}_1(t)$ be the state produced by one dilation step with a single transition operator $K$. Then for every $t \geq 0$
    \begin{align}
        \norm{e^{\mathcal{L}_H t} e^{\mathcal{L}_K t} \rho(0) - \tilde{\rho}_1(t)}_1 \leq \frac{8}{3} t^2 \norm{K}^4 \leq t^2 C_D ,
    \end{align}
    with the dilation constant
    \begin{align}
        C_D := \frac{8}{3} \norm{A}^4 F_\text{tot}^4 . \label{eqn:C_D_def}
    \end{align}
\end{proposition}

\begin{proof}
    Both the exact and the dilated evolution carry the same factor $e^{\Lind_H t}$, which is a conjugation by a unitary and therefore leaves the trace norm invariant, so it suffices to compare the dissipative parts,
    \begin{align}
        \norm{e^{\mathcal{L}_H t} e^{\mathcal{L}_K t} \rho(0) - \tilde{\rho}_1(t)}_1 = \norm{ e^{\mathcal{L}_K t}\rho - \Tr_a \left[ e^{-i \sqrt{t} \tilde{K}} \ket{0} \bra{0} \otimes \rho \, e^{i \sqrt{t} \tilde{K}}  \right]}_1 .
    \end{align}
    Both terms agree with $(\mathbb{1} + t \Lind_K)\rho$ to second order in $t$, so we insert this common expansion and split the error with the triangle inequality,
    \begin{align}
        \norm{e^{\mathcal{L}_H t} e^{\mathcal{L}_K t} \rho(0) - \tilde{\rho}_1(t)}_1 \leq \underbrace{\tnorm{e^{\Lind_K t}\rho - \left( \mathbb{1} + t \Lind_K \right)\rho}}_{(1)} + \underbrace{\tnorm{\Tr_a \left[ e^{-i \sqrt{t} \ad_{\tilde{K}}} \ket{0} \bra{0} \otimes \rho \right] - \left( \mathbb{1} + t \Lind_K \right)\rho}}_{(2)} ,
        \label{eqn:dilation_error_parts}
    \end{align}
    where $\ad_{\tilde{K}}(X) := [\tilde{K},X]$, so that $e^{-i s\, \ad_{\tilde{K}}}(X) = e^{-is \, \tilde{K}} X e^{is \, \tilde{K}}$.

    Term $(1)$ is the integral identity in \cref{sec:duhamel} at order $n=2$,
    \begin{align}
        (1) \leq \frac{t^2}{2} \tnorm{\Lind_K^2[\rho]} \leq \frac{t^2}{2} \norm{\Lind_K}_1^2 \leq 2 t^2 \norm{K}^4 ,
    \end{align}
    using $\norm{\Lind_K}_1 \leq 2\norm{K}^2$ from the proof of \cref{prop:splitting_error}.

    For term $(2)$ we expand the dilated evolution in $\sqrt{t}$ instead. Since $\tilde{K}$ is block off-diagonal in the ancilla while $\ket{0}\bra{0} \otimes \rho$ is block diagonal, $\ad_{\tilde{K}}$ flips this parity, so all odd orders are block off-diagonal and are annihilated by $\Tr_a$. The surviving second-order term is exactly $t \Lind_K[\rho]$, which is what makes the dilation reproduce the Lindbladian, and the first non-vanishing correction is therefore of fourth order in $\sqrt{t}$. Applying \cref{sec:duhamel} at order $n=4$ in $\sqrt{t}$, and using that the partial trace is trace-norm non-increasing,
    \begin{align}
        (2) \leq \frac{(\sqrt{t})^4}{4!} \tnorm{\ad_{\tilde{K}}^4 \left[ \ket{0} \bra{0} \otimes \rho \right]} \leq \frac{t^2}{24} \, 2^4 \norm{\tilde{K}}^4 = \frac{2}{3} t^2 \norm{K}^4 , \label{eqn:dilation_error_term_2}
    \end{align}
    where each of the four nested commutators contributes a factor $2\norm{\tilde{K}}$, and $\norm{\tilde{K}} = \norm{K}$ because $\tilde{K}^2 = \mathrm{diag}(K^\dagger K, K K^\dagger)$.

    Adding the two contributions gives the first inequality, and the second follows by substituting $\norm{K} \leq \norm{A} F_\text{tot}$ as in \cref{prop:splitting_error}.
\end{proof}

In the case of multiple jump operators, $\mathcal{L}_K = \sum_k p_k \mathcal{L}_{K,k}$, we simulate these by classically sampling the corresponding dilation channel, $\bar{\Phi} = \sum_k p_k \Phi^\text{dil}_k$. 
Since $1+t\mathcal{L}_K=\sum_k p_k\left(1+t\mathcal{L}_{K,k}\right)$, the proof above applies unchanged.
Term (1) of \cref{eqn:dilation_error_parts} uses only
$\norm{\mathcal{L}_K}_1\le\sum_k p_k\norm{\mathcal{L}_{K,k}}_1\le 2\max_k\norm{K_k}^2$. And term (2) becomes a convex combination of the single-operator bounds \cref{eqn:dilation_error_term_2}. Hence also no additional error arises beyond replacing $\norm{K} \to \max_k \norm{K_k}$.
Since the sum of the channels and the Lindblad evolution agree to first order, and because the dilation step is a second-order approximation anyway, this introduces no additional trace-norm error besides replacing $\norm{K} \to \max_k \norm{K_k}$.

\subsection{Quadrature error}
\label{sec:quad_error}

Approximating the time integral in the transition operator by the finite sum \cref{eqn:K_S_def} replaces $\tilde{K}$ by $\tilde{K}_S$, we need to bound the difference
\begin{align}
    \norm{\tilde{\rho}_{1} - \tilde{\rho}_{1,s}}_1  &= \norm{ \Tr_a \left( e^{-i \sqrt{\tau} \tilde{K}} \ket{0} \bra{0} \otimes \rho(0) \, e^{i \sqrt{\tau} \tilde{K}} \right) - \Tr_a \left( e^{-i \sqrt{\tau} \tilde{K}_S} \ket{0} \bra{0} \otimes \rho(0) \, e^{i \sqrt{\tau} \tilde{K}_S} \right) } .
\end{align}
This error is controlled by the mass of the filter function outside the sampled time window,
\begin{align}
    F(M_S) := \sum_{|l| \geq M_S} w_l |f(l \tau_S)| ,
\end{align}
with $w_l = \tau_S$, $w_{\pm M_S} = \tau_S/2$. $F(M_S)$, and therefore the error, decays as an exponential in $\sqrt{M_S}$, because $\hat{f}$ is smooth at the edges of its support.

\begin{proposition}
\label{prop:filter_tail}
    Let $\hat{f}$ be the bump filter \cref{eqn:filter_def} with $0 < a \leq 9$, in the convention $f(s) = \frac{1}{2\pi} \int_{-\infty}^{\infty} \hat{f}(\omega) e^{-i \omega s} d\omega$. Then the single-step quadrature error is bounded by
    \begin{align}
        \norm{\tilde{\rho}_{1} - \tilde{\rho}_{1,s}}_1 \leq 4 \tau \norm{A}^2 F_\text{tot} F(M_S) , \label{eqn:quadrature_error}
    \end{align}
    the filter itself obeys the pointwise bound
    \begin{align}
        \left| f(s) \right| \leq \frac{e \, S_\omega}{2\pi} \exp\left( -\frac{1}{6} \sqrt{a S_\omega |s|} \right) \label{eqn:filter_pointwise}
    \end{align}
    for all $s \in \mathbb{R}$. Consequently $F(M_S)$, and with it the quadrature error, is suppressed as $\exp(-\sqrt{a \pi (M_S-1/2)}/6)$.
\end{proposition}

The quadrature enters the dilation step only through the transition operator, so the first task is to convert the difference between the densities into a difference between the transition operators, $\norm{K_S - K}$, by extending the proof in \cite{ding2024single} to include the prefactors:
\begin{proposition}
\label{lem:quad_operator}
    Let $\tilde{\rho}_1$ and $\tilde{\rho}_{1,s}$ be the states produced by a single dilation step of duration $\tau$ with the transition operators $K$ and $K_S$ respectively. Then
    \begin{align}
        \norm{\tilde{\rho}_{1} - \tilde{\rho}_{1,s}}_1 \leq 2 \tau \norm{K - K_S} \left( \norm{K} + \norm{K_S} \right) . \label{eqn:quad_operator}
    \end{align}
\end{proposition}

\begin{proof}
    Conjugation by $e^{-iH\tau}$ is unitary and therefore leaves the trace norm invariant, so writing $\tilde{\rho}(0) := \ket{0}\bra{0} \otimes \rho(0)$ for the initial dilated state and
    \begin{align}
        \tilde{\rho}_{1,a,s}(t) := e^{-i t \tilde{K}_S} \, \tilde{\rho}(0) \, e^{i t \tilde{K}_S}
    \end{align}
    for the quadrature evolution before the ancilla is traced out, we have
    \begin{align}
        \norm{\tilde{\rho}_{1} - \tilde{\rho}_{1,s}}_1 = \norm{\Tr_a \left[ e^{-i \sqrt{\tau} \tilde{K}} \, \tilde{\rho}(0) \, e^{i \sqrt{\tau} \tilde{K}} - \tilde{\rho}_{1,a,s}(\sqrt{\tau}) \right]}_1 .
    \end{align}
    The quadrature evolution obeys $\partial_t \tilde{\rho}_{1,a,s} = -i [ \tilde{K}_S, \tilde{\rho}_{1,a,s} ]$, which we recast as the equation of motion generated by $\tilde{K}$ together with an inhomogeneity,
    \begin{align}
        \partial_t \tilde{\rho}_{1,a,s} = -i \left[ \tilde{K}, \tilde{\rho}_{1,a,s} \right] + i \left[ \tilde{K} - \tilde{K}_S, \tilde{\rho}_{1,a,s} \right] .
    \end{align}
    Both evolutions start from $\tilde{\rho}(0)$, so variation of constants (\cref{sec:duhamel}) yields the exact representation
    \begin{align}
        \norm{\tilde{\rho}_{1} - \tilde{\rho}_{1,s}}_1 = \norm{\Tr_a \left[ \int_0^{\sqrt{\tau}} \! dt' \; e^{-i\tilde{K}(\sqrt{\tau}-t')} \left[ \tilde{K} - \tilde{K}_S, \, \tilde{\rho}_{1,a,s}(t') \right] e^{i\tilde{K}(\sqrt{\tau}-t')} \right]}_1 . \label{eqn:quad_voc}
    \end{align}
    Estimating \cref{eqn:quad_voc} term by term would bound the integrand by a constant and give only $\norm{\tilde{\rho}_{1} - \tilde{\rho}_{1,s}}_1 = O(\sqrt{\tau})$. To exhibit the cancellation we subtract and add the initial state inside the commutator,
    \begin{align}
        \left[ \tilde{K} - \tilde{K}_S, \tilde{\rho}_{1,a,s}(t') \right] = \left[ \tilde{K} - \tilde{K}_S, \, \tilde{\rho}_{1,a,s}(t') - \tilde{\rho}_{1,a,s}(0) \right] + \left[ \tilde{K} - \tilde{K}_S, \, \tilde{\rho}_{1,a,s}(0) \right] ,
    \end{align}
    and in the second term we use the identity $UCU^\dagger-C = (U-\mathbbm{1})C U^\dagger + C(U^\dagger-\mathbbm{1})$. The resulting contribution in which both exponentials are replaced by the identity vanishes: $\tilde{K} - \tilde{K}_S$ is block off-diagonal in the ancilla while $\tilde{\rho}_{1,a,s}(0)$ is block diagonal, so their commutator is block off-diagonal and hence
    \begin{align}
        \Tr_a \left[ \tilde{K} - \tilde{K}_S, \, \tilde{\rho}_{1,a,s}(0) \right] = 0 . \label{eqn:quad_parity}
    \end{align}
    Every surviving contribution therefore carries either the factor $\tilde{\rho}_{1,a,s}(t') - \tilde{\rho}_{1,a,s}(0)$ or a factor $e^{\mp i\tilde{K}(\sqrt{\tau}-t')} - \mathbbm{1}$, each of which supplies the second power of $\sqrt{\tau}$. Using that the partial trace does not increase the trace norm, that $\tnorm{\tilde{\rho}_{1,a,s}} = 1$, and that $\tnorm{[X,Y]} \leq 2\norm{X}\tnorm{Y}$, we obtain
    \begin{align}
        \norm{\tilde{\rho}_{1} - \tilde{\rho}_{1,s}}_1 \leq \;
        & 2 \int_0^{\sqrt{\tau}} \! dt' \, \norm{\tilde{K} - \tilde{K}_S} \, \tnorm{\tilde{\rho}_{1,a,s}(t') - \tilde{\rho}_{1,a,s}(0)} \notag \\
        + \; & 2 \int_0^{\sqrt{\tau}} \! dt' \, \norm{\tilde{K} - \tilde{K}_S} \, \norm{e^{-i\tilde{K}(\sqrt{\tau}-t')} - \mathbbm{1}} \notag \\
        + \; & 2 \int_0^{\sqrt{\tau}} \! dt' \, \norm{\tilde{K} - \tilde{K}_S} \, \norm{e^{i\tilde{K}(\sqrt{\tau}-t')} - \mathbbm{1}} . \label{eqn:quad_three}
    \end{align}
    Both remaining factors are estimated by the same argument applied once more. Differentiating $\tilde{\rho}_{1,a,s}$ and integrating gives
    \begin{align}
        \tilde{\rho}_{1,a,s}(t') - \tilde{\rho}_{1,a,s}(0) = -i \int_0^{t'} \! d\tau_2 \; e^{-i \tau_2 \tilde{K}_S} \left[ \tilde{K}_S, \tilde{\rho}(0) \right] e^{i \tau_2 \tilde{K}_S} ,
    \end{align}
    so that $\tnorm{\tilde{\rho}_{1,a,s}(t') - \tilde{\rho}_{1,a,s}(0)} \leq 2 t' \norm{\tilde{K}_S}$, while $g(u) := e^{-i \tilde{K} u} - \mathbbm{1}$ satisfies $g'(u) = -i\tilde{K}\left( g(u) + \mathbbm{1} \right)$ with $g(0) = 0$, whence
    \begin{align}
        g(u) = -i \int_0^u \! d\tau_2 \; e^{-i\tilde{K}(u - \tau_2)} \tilde{K}
        \qquad \Longrightarrow \qquad \norm{g(u)} \leq u \norm{\tilde{K}} ,
    \end{align}
    evaluated at $u = \sqrt{\tau} - t'$. Inserting both estimates into \cref{eqn:quad_three} and using $\int_0^{\sqrt{\tau}} t' \, dt' = \int_0^{\sqrt{\tau}} (\sqrt{\tau} - t') \, dt' = \tau/2$ gives
    \begin{align}
        \norm{\tilde{\rho}_{1} - \tilde{\rho}_{1,s}}_1 \leq 2 \tau \norm{\tilde{K} - \tilde{K}_S} \left( \norm{\tilde{K}} + \norm{\tilde{K}_S} \right) ,
    \end{align}
    and \cref{eqn:quad_operator} follows because the dilation preserves the spectral norm, $\norm{\tilde{K}} = \norm{K}$.
\end{proof}

\begin{proof}[Proof of \cref{prop:filter_tail}]
    By \cref{lem:quad_operator} it remains to bound $\norm{K - K_S}$, and this is where the tail integral $F(M_S)$ enters. 
In the eigenbasis of $H$ each matrix element of the integrand is
$f(s)e^{i\Delta_{ij}s}A_{ij}$ with $\Delta_{ij}=E_i-E_j$, whose Fourier
transform is $\nu\mapsto\hat f(\nu+\Delta_{ij})$. With $s_l=l\tau_S$,
Poisson summation then gives
\begin{equation}
\tau_S\sum_{l\in\mathbb Z} f(s_l)\,e^{i\Delta_{ij}s_l}
= \sum_{k\in\mathbb Z}\hat f(\Delta_{ij}+2\pi k /\tau_S) = \hat{f}(\Delta_{ij}),
\end{equation}
for $\tau_S \leq \pi/S_\omega$, since $|\Delta_{ij}|\le 2\|H\|\le S_\omega$, so all aliases with $k\neq0$
lie outside $(-S_\omega,0)$, where $\hat f$ vanishes. The infinite sum is
therefore exact,
    \begin{align}
        \tau_S \sum_{l \in \mathbb{Z}} f(s_l) A(s_l) = \int_{-\infty}^{\infty} f(s) A(s) \, ds = K , \label{eqn:poisson}
    \end{align}
    so sampling as such introduces no error for $\tau_S \leq \pi/S_\omega$, and the entire error is the truncation to $|l| \leq M_S$,
\begin{align}
\norm{K - K_S}
= \norm{\left( \sum_{l\in\mathbb Z} \tau_S - \sum_{l=-M_S}^{M_S} w_l \right)\, f(s_l)A(s_l)}
\le \norm{A}\sum_{|l|\ge M_S}w_l\,|f(s_l)| ,
\end{align}
where $w_l=\tau_S$, $w_{\pm M_S} = \tau_S/2$, so $\tau_S - w_l$ vanishes for $|l| < M_S$ and equals $\tau_S/2$ at $l=\pm M_S$ and is equal to $\tau_S$ beyond. Comparing with the decreasing envelope $B$ of \cref{eqn:filter_pointwise},
$\tfrac{\tau_S}{2}B(M_S\tau_S)\le\int_{S_S-\tau_S/2}^{S_S}B$ and
$\tau_S B(l\tau_S)\le\int_{(l-1)\tau_S}^{l\tau_S}B$ for $l>M_S$, hence
\begin{align}
F(M_S) = \sum_{|l|\ge M_S} w_l\,|f(s_l)|
\le \int_{|s|>S_S-\tau_S/2} B(|s|)\, ds ,
\end{align}
   In particular $\norm{K - K_S} \leq \norm{A} F(M_S)$, so combining \cref{lem:quad_operator} with $\norm{K}, \norm{K_S} \leq \norm{A} F_\text{tot}$, which holds for both operators as in \cref{prop:splitting_error}, gives
    \begin{align}
        \norm{\tilde{\rho}_{1} - \tilde{\rho}_{1,s}}_1 \leq 2 \tau \cdot \norm{A} F(M_S) \cdot 2 \norm{A} F_\text{tot} = 4 \tau \norm{A}^2 F_\text{tot} F(M_S) ,
    \end{align}
    which is \cref{eqn:quadrature_error}.

    It remains to prove \cref{eqn:filter_pointwise} and to evaluate the envelope integral. Set $g(x) := e^{-a/(1-x^2)}$ for $|x| < 1$ and $g := 0$ otherwise, so that $g \in C^\infty(\mathbb{R})$ with $\text{supp} \, g = [-1,1]$ and $\hat{f}(\omega) = g(2\omega/S_\omega + 1)$. Substituting $x = 2\omega/S_\omega + 1$ into the definition of $f$ gives
    \begin{align}
        f(s) = \frac{S_\omega}{4\pi} e^{i S_\omega s/2} \, \hat{g}\!\left( \frac{S_\omega s}{2} \right) , \qquad \hat{g}(\xi) := \int_{-1}^{1} g(x) e^{-i \xi x} dx ,
    \end{align}
    so that \cref{eqn:filter_pointwise} is equivalent to $|\hat{g}(\xi)| \leq 2 e \, e^{-\sqrt{a|\xi|/18}}$ for $a|\xi| \geq 18$, which we prove in four steps.

    Since $g$ is smooth with compact support, all boundary terms vanish under repeated integration by parts, so for every $n \geq 1$
    \begin{align}
        \left| \hat{g}(\xi) \right| \leq \frac{\norm{g^{(n)}}_{L^1}}{|\xi|^n} \leq \frac{2 \sup_x |g^{(n)}(x)|}{|\xi|^n} . \label{eqn:filter_ibp}
    \end{align}
    Since $g$ vanishes identically outside $[-1,1]$ and all its derivatives vanish at $\pm 1$, we have $g^{(n)}(x) = 0$ for $|x| \geq 1$, so the supremum in \cref{eqn:filter_ibp} is attained on the open interval $(-1,1)$; since $g$ is even it is enough to bound it for $x \in [0,1)$.

    Next, $G(z) := \exp(-a/(1-z^2))$ is analytic away from $z = \pm 1$. Fix $x \in [0,1)$, set $\delta := 1 - x \in (0,1]$ and consider the circle $|z - x| = \delta/2$, which stays in the domain of analyticity because $|z-1| \geq \delta/2 > 0$ and $|z+1| \geq 1 - \delta/2 > 0$ there. On this circle $\operatorname{Re}(1-z) \geq \delta/2$ and $|1-z| \leq 3\delta/2$, so that
    \begin{align}
        \operatorname{Re} \frac{1}{1-z} = \frac{\operatorname{Re}(1-z)}{|1-z|^2} \geq \frac{\delta/2}{(3\delta/2)^2} = \frac{2}{9\delta} ,
    \end{align}
    while $\operatorname{Re}(1+z) \geq 1 - \delta/2 > 0$ implies $\operatorname{Re}(1+z)^{-1} > 0$. The partial fraction decomposition $(1-z^2)^{-1} = \frac{1}{2}[(1-z)^{-1} + (1+z)^{-1}]$ therefore gives $\operatorname{Re}(1-z^2)^{-1} \geq 1/(9\delta)$, hence $|G(z)| = e^{-a \operatorname{Re}(1-z^2)^{-1}} \leq e^{-a/(9\delta)}$ on the circle, and the Cauchy estimate yields
    \begin{align}
        \left| g^{(n)}(x) \right| \leq n! \left( \frac{2}{\delta} \right)^{n} e^{-a/(9\delta)} .
    \end{align}
    The right-hand side is maximized over $\delta \in (0,1]$ at $\delta = a/(9n)$, which is admissible because $a \leq 9 \leq 9n$, giving
    \begin{align}
        \sup_x \left| g^{(n)}(x) \right| \leq n! \left( \frac{18 \, n}{e \, a} \right)^{n} . \label{eqn:filter_gevrey}
    \end{align}
    Finally, inserting \cref{eqn:filter_gevrey} into \cref{eqn:filter_ibp}, using $n! \leq n^n$ and writing $Q := a|\xi|/18$,
    \begin{align}
        \left| \hat{g}(\xi) \right| \leq 2 \, n! \left( \frac{18 \, n}{e \, a |\xi|} \right)^{n} \leq 2 \left( \frac{n^2}{e \, Q} \right)^{n} .
    \end{align}
    Choosing $n = \lfloor \sqrt{Q} \rfloor$, which is at least one precisely when $a |\xi| \geq 18$, gives $n^2 \leq Q$, so the bracket is bounded by $e^{-1}$ and $|\hat{g}(\xi)| \leq 2 e^{-n}$. Undoing the floor with $n = \lfloor \sqrt{Q} \rfloor > \sqrt{Q} - 1$ then gives $|\hat{g}(\xi)| \leq 2 e^{1 - \sqrt{Q}}$, which is the claim.

    Finally we evaluate the envelope integral. Since $|f|$ is even, with $\beta := \sqrt{a S_\omega}/6$ and the substitution $s = v^2$,
    \begin{align}
        \int_{|s|>S_S-\tau_S/2} B(|s|) \, ds = \frac{e S_\omega}{\pi} \int_{|s|>S_S-\tau_S/2}^{\infty} e^{-\beta \sqrt{s}} ds = \frac{e S_\omega}{\pi} \cdot \frac{2}{\beta^2} \left( 1 + \beta \sqrt{S_S-\tau_S/2} \right) e^{-\beta \sqrt{S_S-\tau_S/2}} ,
    \end{align}
    and inserting $\beta^2 = a S_\omega / 36$ together with $S_\omega S_S = M_S \pi$ gives,
    \begin{align}
        F(M_S) \leq \int_{|s|>S_S-\tau_S/2} B(|s|) \, ds = \frac{72 \, e}{\pi a} \left( 1 + \frac{\sqrt{a \pi (M_S-1/2)}}{6} \right) \exp\left( - \frac{\sqrt{a \pi (M_S-1/2)}}{6} \right) . \label{eqn:filter_tail}
    \end{align}
For $aS_\omega|s|<36$ the bound holds trivially: since
$0\le\hat f\le e^{-a}$ on an interval of length $S_\omega$,
$|f(s)|\le\frac{1}{2\pi}\int\hat f\,d\omega\le S_\omega/2\pi$, while
$\tfrac16\sqrt{aS_\omega|s|}<1$ implies $B(|s|)>S_\omega/2\pi$.
Hence $|f|\le B$ everywhere, and since $B$ is decreasing it holds for all $M_S$.
\end{proof}

Since $F(M_S)$ depends only on $a$ and $M_S$ it is system independent, so in practice we evaluate it numerically, which leads to much tighter bounds. \Cref{prop:filter_tail} only certifies the $e^{-\sqrt{a M_S}}$ scaling quoted in \cref{eqn:eps_quad_scaling}.

\subsection{Quadrature Trotterization error}
\label{sec:quad_trott_error}
\begin{proposition}
\label{prop:quad_trott_error}
    The single-step error due to the Trotterization of the discretized time integral for time step $\tau$ is bounded by
    \begin{align}
        \norm{\tilde{\rho}_{1,s}-\bar{\rho}_1}_1 \leq \tau^2 C_{T,1} + \tau^3C_{T,2} ,
    \end{align}
    with
    \begin{align}
    \begin{split}
    C_{T,1} &:=  \sum_{k=1}^{K}  \left[ \frac{1}{24} \norm{\ad^3_{V_{k}}(v_{k})} + \frac{1}{96} \norm{\ad^3_{v_{k}}( V_{k})} \right. \\
        & \quad \left. + \frac{1}{24} \Big\| \ad^2_{V_k} v_k
        -  \frac{1}{2}\ad^2_{v_k} V_k \Big\|
        \Big( 11 \norm{V_{k+1}} + 3 \norm{V_k} + 3 \norm{v_k} + 8\sum_{k'=k+1}^K \norm{v_{k'}} \Big)\right] ,
    \end{split}\\
        C_{T,2} &:= \sum_{k=1}^{K}  \frac{1}{144} \left( \norm{\ad^2_{V_{k}} (v_{k})} +  \frac{1}{2} \norm{\ad^2_{v_{k}}(V_{k})} \right)^2 .
    \end{align}
    where we used the alternative notation 
    \begin{align}
        v_k :=&  w_{k-1-M_S} (\sigma_x \mathrm{Re} f_{k-1-M_S} +\sigma_y \mathrm{Im} f_{k-1-M_S})\otimes A(s_{k-1-M_S}) \quad K = 2 M_S +1 ,\\
        V =& \sum_{k=1}^{K} v_k, \quad V_{k} = \sum_{k'=1}^{k-1} v_{k'} .
    \end{align}
    The total error from this source over $M=T/\tau$ steps is then 
    \begin{align}
           \norm{\tilde{\rho}_{M,s}-\bar{\rho}_{M}}_1 \leq T (\tau C_{T,1} + \tau^2 C_{T,2}).\label{eqn:quad_trott_error}
    \end{align}
\end{proposition}
Here the adjoint denotes the commutator,
\begin{align}
    \ad_A(B) = [A,B] .
\end{align}

\begin{proof}
We are looking to bound $ \norm{\tilde\rho_{1,s} - \bar\rho_1}_1$. After absorbing the coherent time evolution into the norm we are left with
\begin{align}
    \norm{\tilde\rho_{1,s} - \bar\rho_1}_1
    = \Big\| \Tr_a \Big[ e^{-itV} \, \tilde\rho \, e^{itV}
        - \Big( \prod_{k=K}^{1} e^{-i \frac{t}{2} v_k} \prod_{k=1}^{K} e^{-i \frac{t}{2} v_k} \Big)
        \, \tilde\rho \,
        \Big( \prod_{k=K}^{1} e^{i \frac{t}{2} v_k} \prod_{k=1}^{K} e^{i \frac{t}{2} v_k} \Big)
        \Big] \Big\|_1 .
    \label{eq:tr_p31f_1}
\end{align}

\paragraph{Telescope over partial product formulas.}
We define 
\begin{align}
    S_k(t) :=
        \Big( \prod_{k'=K}^{k} e^{-i\frac{t}{2}v_{k'}} \Big) \,
        e^{-it V_k} \,
        \Big( \prod_{k'=k}^{K} e^{-i\frac{t}{2}v_{k'}} \Big) ,
    \qquad S_1 = \prod_{k=K}^{1} e^{-i \frac{t}{2} v_k} \prod_{k=1}^{K} e^{-i \frac{t}{2} v_k} , \quad S_{K+1} = e^{-itV} ,
\end{align}
and split the Trotterization into K individual steps by inserting  $\pm \sum_{k=2}^{K} S_k \, \tilde\rho \, S^\dagger_k$:
\begin{align}
    \norm{\tilde\rho_{1,s} - \bar\rho_1}_1
    = \Big\| \sum_{k=1}^{K} \Tr_a \Big[
        S_k(t) \, \tilde\rho \, S^\dagger_k
        - S_{k+1}(t) \, \tilde\rho \, S^\dagger_{k+1}
        \Big] \Big\|_1 .
\end{align}

\paragraph{ Add traceless terms to ensure $O(t^4)$.}
The second-order Trotter formula is $O(t^3)$, but we gain an additional order because the leading term in $t$ is traceless. In the proof we capture this by inserting the term $D_k$ defined as
\begin{align}
    D_k := \int_0^{t} \! dt_1 \, \frac{t_1^2}{2}
        \left[ \ad^2_{-iV_k} \Big( -\frac{i}{2} v_k \Big)
             + \ad^2_{i/2 v_k} \big( i V_k \big) \right] ,
    \label{eq:D_def}
\end{align}
which is off-diagonal and therefore stays traceless when multiplied by the dilated density matrix:
\begin{align}
    \norm{\tilde\rho_{1,s} - \bar\rho_1}_1 = \Big\| \sum_{k=1}^{K} \Tr_a \Big[
        S_k \, \tilde\rho \, S^\dagger_k
        - S_{k+1} \, \tilde\rho \, S^\dagger_{k+1}
        - D_k \tilde\rho - \tilde\rho \, D^\dagger_k \Big] \Big\|_1 .
\end{align}

\paragraph{ Remove the trace and rotate.}
We use $\norm{\Tr_a [\,\cdots]}_1 \leq \norm{[\,\cdots]}_1$, pull the norm into the sum, and
rotate the $k$-th term by unitary transformation $U_k := \prod_{k'=k+1}^{K} e^{i\frac{t}{2}v_{k'}}$ from the
left and $U_k^\dagger$ from the right. Absorbing the remaining unitaries into the density matrix
\begin{align}
    \tilde\rho_k(t) := \left( \prod_{k'=k+1}^{K} e^{-i\frac{t}{2}v_{k'}} \right)\;
        \tilde\rho \; \left( \prod_{k'=K}^{k+1} e^{+i\frac{t}{2}v_{k'}} \right)
\end{align}
we arrive at
\begin{align}
    \norm{\tilde\rho_{1,s} - \bar\rho_1}_1 \leq \sum_{k=1}^{K} \Big\|
        S_{2,k} \, \tilde\rho_k(t) \, S^\dagger_{2,k}
      - e^{-i t V_{k+1}} \, \tilde\rho_k(t) \, e^{i t V_{k+1}}
      - U_k \big( D_k \tilde\rho + \tilde\rho \, D^\dagger_k \big) U_k^\dagger
      \Big\|_1 ,
\end{align}
where we introduced two-term product formula for each splitting $V_{k+1} = V_k + v_k$
\begin{align}
    S_{2,k}(t) &:= e^{-i\frac{t}{2}v_k} \, e^{-i t V_k} \, e^{-i\frac{t}{2}v_k} .
\end{align}
We use the known result for the second-order Trotter formula from Childs et al. \cite{childs2021theory}, equation  (L 4):
\begin{align}
    S_{2,k}(t) &= e^{-i t V_{k+1}} + A_k(t) , \\
    A_k(t) &= \int_0^{t} \!\! dt_1 \int_0^{t_1} \!\! dt_2
        \int_0^{t_2} \!\! dt_3 \;
        e^{-i(t-t_1)V_{k+1}} \, e^{-i\frac{t_1}{2}v_k} \nonumber \\
    &\qquad \times \left[ e^{-it_3 \ad_{V_k}} \ad^2_{-i V_k}
                \Big( -\frac{i}{2}v_k \Big)
             + e^{i\frac{t_3}{2}\ad_{v_k}} \ad^2_{\frac{i}{2}v_k}
                \big( i V_k \big) \right] e^{-it_1V_k} e^{-i t_1/2 v_k}.
\end{align}
Inserting this the square terms cancel, the two cross terms are adjoints of each other, and $\norm{\tilde\rho_k(t)}_1 = 1$, so we are left with
\begin{align}
    \norm{\tilde\rho_{1,s} - \bar\rho_1}_1 \leq \sum_{k=1}^{K} \Big[
        2 \Big\| A_k(t) \, \tilde\rho_k(t) \, e^{i t V_{k+1}}
        - U_k D_k \tilde\rho \, U_k^\dagger \Big\|_1
        + \norm{A_k(t)}^2 \Big] .
\end{align}

\paragraph{ Split $A_k$.}
Writing the integrand of $A_k$ as its value at $t_3 = 0$ plus the integral of its
derivative gives $A_k = A^{(1)}_k + A^{(2)}_k$ with
\begin{align}
    A^{(1)}_k(t) &= \int_0^{t}\!\! dt_1 \int_0^{t_1}\!\! dt_2
        \int_0^{t_2}\!\! dt_3 \int_0^{t_3}\!\! dt_4 \;
        e^{-i(t-t_1)V_{k+1}} e^{-i\frac{t_1}{2}v_k} \nonumber \\
    &\qquad \times \left[ e^{-it_4 \ad_{V_k}} \ad^3_{-i  V_k}
            \Big( -\frac{i}{2}v_k \Big)
        + e^{i\frac{t_4}{2}\ad_{v_k}} \ad^3_{\frac{i}{2}v_k}
            \big( i  V_k \big) \right]  e^{-it_1V_k} e^{-i t_1/2 v_k}, \\
    A^{(2)}_k(t) &= e^{-itV_{k+1}} \int_0^{t}\!\! dt_1 \, \frac{t_1^2}{2} \, T(t_1) , \\
    T(t) &:= e^{itV_{k+1}} \, e^{-i\frac{t}{2}v_k}
        \left[ \ad^2_{-i  V_k} \Big( -\frac{i}{2}v_k \Big)
        + \ad^2_{\frac{i}{2}v_k} \big( i  V_k \big) \right]
        e^{-i t V_k} \, e^{-\frac{i}{2}t v_k} .
\end{align}
$A^{(1)}_k = O(t^4)$, so $A^{(2)}_k = O(t^3)$ has the leading order time dependence, which will cancel against the corresponding term in $U_k D_k \tilde\rho \, U_k^\dagger$:
\begin{equation}
    \begin{aligned}
    \norm{\tilde\rho_{1,s} - \bar\rho_1}_1 \leq \sum_{k=1}^{K} \Big[
    & \; \underbrace{2 \norm{A^{(1)}_k(t) \, \tilde\rho_k(t) \, e^{i t V_{k+1}}}_1}_{\circled{1}}
    + \underbrace{2 \Big\| A^{(2)}_k(t) \, \tilde\rho_k(t) \, e^{i t V_{k+1}}
        - U_k D_k \tilde\rho \, U_k^\dagger \Big\|_1}_{\circled{2}}
    + \underbrace{\norm{A_k(t)}^2}_{\circled{3}} \Big] .
    \end{aligned}
\end{equation}

\paragraph{ Terms $\circled{1}$ and $\circled{3}$.}
These terms can be evaluated by using the triangle inequality, pulling the norm under the integrals and using the submultiplicativity of the trace norm (\ref{eqn:submultiplicativity}), which offloads the trace norm onto the density matrix, leaving us with the following spectral norm bounds:
\begin{align}
    \circled{1} &\leq 2 \cdot \frac{t^4}{24}
        \left( \norm{\ad^3_{V_k} \frac{v_k}{2}} 
             + \norm{\ad^3_{\frac{v_k}{2}} V_k} \right)
    = \frac{t^4}{24} \left( \norm{\ad^3_{V_k} v_k}
             + \frac{1}{4} \norm{\ad^3_{v_k} V_k} \right) , \\
    \circled{3} &\leq \left[ \frac{t^3}{12} \norm{\ad^2_{V_k} v_k}
        + \frac{t^3}{24} \norm{\ad^2_{v_k} V_k} \right]^2
    = \frac{t^6}{144} \left[ \norm{\ad^2_{V_k} v_k}
        + \frac{1}{2} \norm{\ad^2_{v_k} V_k} \right]^2 .
\end{align}

\paragraph{ Term $\circled{2}$.}
Using the definition of $D_k$ (\ref{eq:D_def}) we can express it in terms of $T(t)$
\begin{align}
    U_k D_k \tilde\rho \, U_k^\dagger = \int_0^t \!\! dt_1 \, \frac{t_1^2}{2}
    \prod_{k'=k+1}^{K} e^{i\frac{t}{2}\ad_{v_{k'}}} \big( T(0) \tilde\rho \big) .
\end{align}
Inserting this and adding a zero we obtain the terms 
\begin{align}
    \circled{2} = 2 \Big\| \int_0^{t} \!\! dt_1 \, \frac{t_1^2}{2} \Big[ &
        \underbrace{e^{-i t \ad_{V_{k+1}}}
            \Big( \big( T(t_1) - T(0) \big) \tilde\rho_k(t) \Big)}_{\circled{2.1}}
        \nonumber \\
    &\qquad + \underbrace{e^{-i t \ad_{V_{k+1}}} \big( T(0) \tilde\rho_k(t) \big)
        - \prod_{k'=k+1}^{K} e^{i\frac{t}{2}\ad_{v_{k'}}}
            \big( T(0) \tilde\rho \big)}_{\circled{2.2}} \Big] \Big\|_1 .
\end{align}
For $\circled{2.1}$, let $f(t) := T(t) - T(0)$, $f(0) = 0$. Writing $f(t) = \int_0^t f'(t_1) \, dt_1$ and pulling the norm under the integral gives:
\begin{align}
    \circled{2.1} &\leq 2 \int_0^{t} \!\! dt_1 \, \frac{t_1^2}{2}
        \int_0^{t_1} \!\! dt_2 \, \Big\| \frac{1}{2} \ad^2_{V_k} v_k
        - \frac{1}{4} \ad^2_{v_k} V_k \Big\|
        \Big( \norm{V_{k+1}} + \norm{V_{k}}
        +\norm{v_k}
        \Big) \\
        &= \frac{t^4}{16} \Big\| 2 \ad^2_{V_k} v_k
        - \ad^2_{v_k} V_k \Big\| \Big( \norm{V_{k+1}} + \norm{V_{k}}
        +\norm{v_k}
        \Big).
\end{align}

For $\circled{2.2}$ it is again easy to see that the integrand vanishes for
$t=0$, so we can also rewrite it in the form $f(t) = \int_0^t f'(t_1) \, dt_1$.
Pulling the norm under the integral and into the commutators gives
\begin{align}
    \circled{2.2} &\leq \frac{t^4}{6} \Big\| 2 \ad^2_{V_k} v_k
        - \ad^2_{v_k} V_k \Big\| \left(\norm{V_{k+1}} + \sum_{k'=k+1}^K \norm{v_{k'}} \right) .
\end{align}

\paragraph{Putting it all together}
Combining all the terms we arrive at the bound:
\begin{equation}
    \begin{aligned}
    \norm{\tilde\rho_{1,s} - \bar\rho_1}_1 \leq \sum_{k=1}^{K} \Big[
    & \; \frac{t^4}{24} \Big( \norm{\ad^3_{V_k} v_k}
        + \frac{1}{4}\norm{\ad^3_{v_k} V_k} \Big)
        + \frac{t^6}{144} \Big[ \norm{\ad^2_{V_k} v_k}
        + \frac{1}{2} \norm{\ad^2_{v_k} V_k} \Big]^2 \\
    + & \; \frac{t^4}{24} \Big\| \ad^2_{V_k} v_k
        -  \frac{1}{2}\ad^2_{v_k} V_k \Big\|
        \Big( 11 \norm{V_{k+1}} + 3 \norm{V_k} + 3 \norm{v_k} + 8\sum_{k'=k+1}^K \norm{v_{k'}} \Big)\Big] .
    \end{aligned}
\end{equation}
\end{proof}

\subsection{$H$ evolution Trotter error}
\label{sec:H_trott_err}
$H$ evolution enters the algorithm both through the coherent evolution and to implement the filter integral, with a total $H$ evolution time  $T_{\text{unitary}} = T + 2 M (2M_S+1) \tau_S = \tilde{O}(T^2)$, with $M_S\cdot M \cdot \tau_S $-time evolution dominating the overall error.

\begin{proposition}
    \label{prop:H_trotter_error}
    The Trotter error for the $H$ evolution over M Lindblad steps is bounded by
    \begin{align}
    \varepsilon_{H,\text{trott}} &= 2 T \tau_{T,1}^2 C_H + 4 M (2M_S+1) \tau_S \tau_{T,2}^2 C_H , \label{eqn:H_trott_err_scaling}
    \end{align}
    where $C_H$ is the error constant for the second-order Trotter formula and $\tau_{T,1}$, $\tau_{T,2}$ are the Trotter step sizes for the coherent and dissipative parts respectively.
\end{proposition}

\begin{proof}
    If we perform coherent time evolution then this gives us for each time step $\tau$:
\begin{align}
    \norm{e^{-iH \tau} \rho e^{iH \tau} - \delta(\tau) \rho \delta^\dagger(\tau)}_1 \leq 2 \norm{e^{-iH \tau}-\delta(\tau)} ,
\end{align}
where $\delta$ is the chosen product formula \cite{childs2021theory}.
We also use unitary time evolution to implement the dissipative part of the algorithm, where we find
\begin{align}
    \norm{\mathcal{W}(\tau)[\rho] - \bar{\mathcal{W}}(\tau)[\rho]}_1 \leq 4(2M_S+1) \norm{e^{iH \tau_S}-\delta(\tau_S)} .
\end{align}
Using the Trotter steps $\tau_{T,1} = \frac{\tau}{M_{T,1}}$ and $\tau_{T,2} = \frac{\tau_S}{M_{T,2}}$ for the coherent and dissipative parts respectively, we obtain for the Trotter error accumulated over $M$ steps for a system Hamiltonian split into $\Gamma$ parts,
\begin{align}
    H = \sum_{\gamma=1}^{\Gamma} H_\gamma ,
\end{align}
by the standard result for the second-order Trotter formula \cite{childs2021theory}:
\begin{align}
    \varepsilon_{H,\text{trott}} &= 2 M \norm{e^{-iH \tau}-\delta(\tau)}  + 4 M (2M_S+1) \norm{e^{iH \tau_S}-\delta(\tau_S)} \\
    &= 2 T \tau_{T,1}^2 C_H + 4 M (2M_S+1) \tau_S \tau_{T,2}^2 C_H ,
\end{align}
where $C_H$ is the error constant
\begin{align}
C_H := \frac{1}{12} \sum_{\gamma_1=1}^{\Gamma} \left\| \left[ \sum_{\gamma_3=\gamma_1+1}^{\Gamma} H_{\gamma_3}, \left[ \sum_{\gamma_2=\gamma_1+1}^{\Gamma} H_{\gamma_2}, H_{\gamma_1} \right] \right] \right\| + \frac{1}{24} \sum_{\gamma_1=1}^{\Gamma} \left\| \left[ H_{\gamma_1}, \left[ H_{\gamma_1}, \sum_{\gamma_2=\gamma_1+1}^{\Gamma} H_{\gamma_2} \right] \right] \right\|. \label{eqn:H_trott_const}
\end{align}
\end{proof}

For the case of the Hubbard model with a 2-partition tiling of the hopping term we can upper-bound the Trotter constant with the results for the 2D square lattice from \cite{campbell2022early}, $C_H \le W_\text{split} + W_\text{tiling}$. Here $W_\text{split}$ refers to the Trotter constant of splitting the Hamiltonian into hopping and interaction terms,
\begin{align}
    W_\text{split} := \frac{U^2t_\text{hop}}{24} \norm{R}_1 +\frac{U t_\text{hop}^2}{6} n_\text{sites} (\sqrt{5}+8) ,
\end{align}
where $R$ is the adjacency matrix of the hopping terms. $W_\text{tiling}$ is the Trotter error constant for splitting the hopping terms into two sets of $S_1$ tiles,
\begin{align}
    W_\text{tiling} := \frac{t_\text{hop}^3}{12} \norm{\left[ [R_1,R_2],R_2 \right]}_1 + \frac{t_\text{hop}^3}{24} \norm{\left[ [R_1,R_2],R_1 \right]}_1 ,
\end{align}
with $R_1$ and $R_2$ being the adjacency matrices of the two hopping partitions. Tighter bounds on the worst-case Trotter error are possible with numerical estimators for this model \cite{bay2025fault, blunt2025monte}; since the empirical parameter choice already takes care of this error, we do not pursue these improvements here.

\subsection{$\tilde{A}$ Trotter error}
\label{sec:A_trott_err}

\begin{proposition}
\label{prop:A_trott_err}
    For first-order Trotter splitting for $\tilde{A}$ we find for the $M$-step error
\begin{align}
    \varepsilon_{A, \text{trott}} =& \tau \tau_S^2 M C_{\tilde{A}} ,
\end{align}
with 
\begin{align}
C_{\tilde{A}} :=  \norm{A}^2 \sum_{l=-M_S}^{M_S} |Imf(s_l)||Ref(s_l)| .
\end{align}
\end{proposition}

\begin{proof}
    The proof proceeds similarly to the filter-integral contribution to the $H$ evolution error:
    \begin{align}
    \varepsilon_{A,\text{trott}} = \norm{\mathcal{W}(\tau)[\rho] - \bar{\mathcal{W}}(\tau)[\rho]}_1 \leq 4\sum_{l=-M_S}^{M_S} \norm{\tilde{A}_l(\sqrt{\tau})-\delta_1(\tau)} ,
\end{align}
with 
\begin{align}
    \tilde{A}(\sqrt{\tau}) &= e^{-i \frac{\sqrt{\tau}}{2} \sigma_l \otimes A} .
\end{align}
Using the first-order Trotter splitting for $\tilde{A}$ \cite{childs2021theory} and summing over all $M$ steps, we obtain the stated error bound:
\begin{align}
    \varepsilon_{A,\text{trott}} =& \tau \tau_S^2 M \norm{A}^2 \sum_{l=-M_S}^{M_S} |\mathrm{Im}f(s_l)||\mathrm{Re}f(s_l)| .
\end{align}
\end{proof}

\subsection{Fixed-point error}
\label{sec:fixed_point_error}

The Lindblad evolution has an associated mixing time, after which any two initial states become similar:
\begin{align}
    \tnorm{e^{\Lind t^{\mathcal{L}}_{\text{mix}}}[\rho-\sigma]} \leq  \frac{1}{2} \tnorm{\rho-\sigma}.
\end{align}
This implies that after some $T^{\mathcal{L}}_{\text{mix}}$ any input state is arbitrarily close to a fixed point of the dynamics. However, one can also directly define the dilated and split channel
\begin{align}
    \mathcal{D}[\rho] := e^{-iH\tau} \Tr_a\!\left[ e^{-i\sqrt{\tau}\tilde K}
    \left( \ket{0}\bra{0} \otimes \rho \right) e^{i\sqrt{\tau}\tilde K} \right] e^{iH\tau} .
\end{align}

As the authors of \cite{ding2024single} note in Appendix C, this channel also fixes the ground state, so under the usual ergodicity and uniqueness assumptions the ground state remains the fixed point of $\mathcal{D}$ with mixing time $T^{D}_{\text{mix}}$. To measure the distance between the ideal ground state and the state reached after $M$ channel applications we can neglect the dilation and Lindblad-splitting errors, although this means that we are no longer bounding the distance to the exact Lindblad evolution, which can mean that $T^{D}_{\text{mix}}$ and $T^{\mathcal{L}}_{\text{mix}}$ may differ. For the rigorous scaling we do not make use of such a simplification, but it is used in the empirical regimes, since in the simulation we can directly evaluate the fixed-point ground-state energy, and not have to control the trace distance to the ideal Lindbladian evolution.

\section{Empirical spectral properties and time steps}
\label{sec:spectral_scaling}

We estimate the spectral properties via power law fitting from small systems, see \cref{fig:spectral_gap_scaling} and \cref{fig:ground_state_energy_scaling}. These then inform the choice of parameters and expected performance for larger systems. Simulations are for 1D chains with open boundary conditions. We obtain exact results up to $10$ sites and then extrapolate. Since these systems are one-dimensional, it would be simple to extend this data using DMRG simulations, but we did not do that here.

We note that a linear relationship between the number of sites and the spectral norm also holds true for more general lattices. The spectral gap, in contrast, depends strongly on the dimension of the lattice. Beyond the crossover point where the extensive error budget $0.005 \cdot n_\text{sites}$ exceeds the 1D gap (around 28 sites for the fit in Fig. 10), $\Delta_\text{eff}$ is set by the error budget, so a smaller gap leaves the cost unchanged and a larger gap can only reduce it. Below the crossover, a smaller gap in higher dimensions would increase the cost.

\begin{figure}
    \centering
    \includegraphics{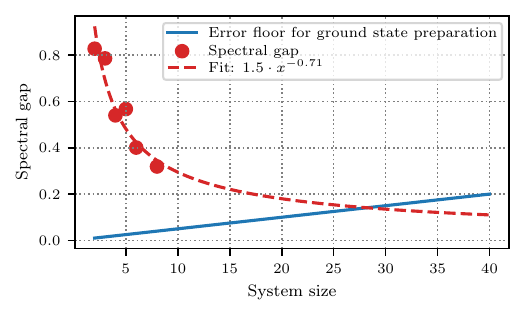}
    \caption{Scaling of the spectral gap with system size for one-dimensional Hubbard models with $U=4$ at half-filling with open boundary conditions. The dashed line shows a fit to this data, used to extrapolate to larger systems. The solid line shows an extensive target accuracy for the energy \cite{campbell2022early}.}
    \label{fig:spectral_gap_scaling}
\end{figure}

\begin{figure}
    \centering
    \includegraphics{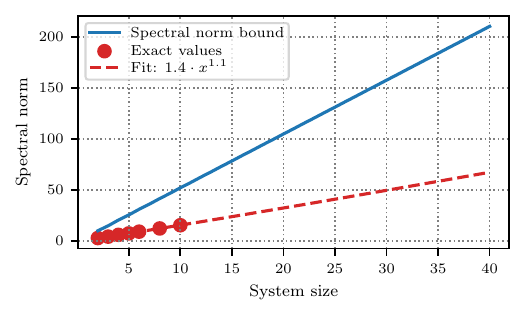}
    \caption{Scaling of the spectral norm, $\norm{H}$, for one-dimensional Hubbard models with $U=4$ and open boundary conditions.}
    \label{fig:ground_state_energy_scaling}
\end{figure}

Finally, we test our prediction that in the empirical regimes the required step size $\tau$ scales inversely proportionally to the mixing time, $T_\text{mix}$, which we can read off for the simulated systems, while being independent of the size of the system. \Cref{fig:time_step_scaling} provides some evidence for this, finding an almost linear relationship between $1/\tau$ and $T_\text{mix}$, while the size of the system quadrupled. From the fit we estimate that the time step required for a single unit of time, $\tau \cdot T_\text{mix}$, will be no smaller than $\tau = 0.15$ for the general scaling and $\tau = 0.25$ for the scaling with known gap.

However, such small systems can only be indicative of the larger lattices that we aim to cost, so it is possible that the step size required for a single time unit does not stay constant with system size, so in this sense our estimates are more of a lower bound to the actual cost. 

\begin{figure}
    \centering
    \includegraphics{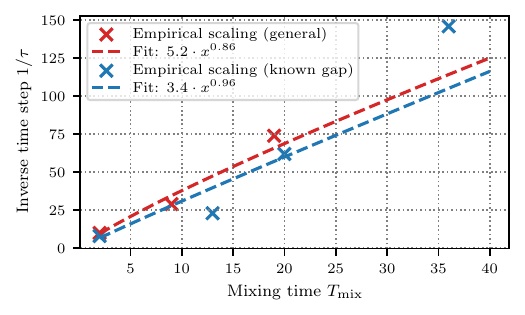}
    \caption{Scaling of the required time step $\tau$ with the mixing time $T_\text{mix}$ for one-dimensional, $U=4$ Hubbard models with open boundary conditions at half-filling.}
    \label{fig:time_step_scaling}
\end{figure}

\end{document}